%% file: main.tex
\documentclass[journal]{IEEEtran}
\usepackage{graphicx} 
\graphicspath{{figure/}}
\usepackage{float}
\usepackage{amsmath,amssymb,amsthm,bm}
\usepackage{xcolor}
\usepackage{array}
\usepackage[caption=false,font=footnotesize]{subfig}
\usepackage{cite} 
\usepackage{mathtools}

\newcommand{\diag}[1]{\ensuremath{\mathrm{diag}(#1)}}

\usepackage{algorithm,algorithmic}

\newtheoremstyle{bfnote}%
{}{}%
{\itshape}{}%
{\bfseries}{.}%
{ }%
{\thmname{#1}\thmnumber{ #2}\thmnote{ (#3)}}
\theoremstyle{bfnote}
\newtheorem{thm}{Theorem}

\newtheorem{lem}{Lemma}

\DeclarePairedDelimiterX\Set[2]{\lbrace}{\rbrace}%
{ #1 \,\delimsize| \,\mathopen{} #2 }

\usepackage{bookmark}
\usepackage{color}
\newcommand\highlight[1]{\textcolor{black}{#1}}
\newcommand\highlights[1]{\textcolor{black}{#1}}

\newcommand{\real}[0]{\mathbb R}

\ifCLASSINFOpdf
\else
\fi
\begin{document}
\bibliographystyle{IEEEtran}
\title{Reinforcement Learning for Optimal Primary Frequency Control: A Lyapunov Approach}

\author{Wenqi Cui, Yan Jiang, and Baosen Zhang
\thanks{The authors are with the Department
of Electrical and Computer Engineering, University of Washington, Seattle, WA, 98195 e-mail:\{wenqicui,jiangyan, zhangbao\}@uw.edu}
\thanks{The authors are supported in part by the National Science Foundation grant ECCS-1930605, ECCS-1942326 and the Washington Clean Energy Institute.}}


\maketitle

\begin{abstract}
As more inverter-connected renewable resources are integrated into the grid, frequency stability may degrade because of the reduction in mechanical inertia and damping. A common approach to mitigate this degradation in performance is to use the power electronic interfaces of the renewable resources for primary frequency control. Since inverter-connected resources can realize almost arbitrary responses to frequency changes, they are not limited to reproducing the linear droop behaviors. To fully leverage their capabilities, reinforcement learning (RL) has emerged as a popular method to design nonlinear controllers to optimize a host of objective functions. 

Because both inverter-connected resources and synchronous generators would be significant part of the grid in the near and intermediate future, the learned controller of the former should be stabilizing with respect to the nonlinear dynamics of the latter. To overcome this challenge, we explicitly engineer the structure of neural network--based controllers such that they guarantee system stability by construction, through the use of a Lyapunov function. A recurrent neural network architecture is used to efficiently train the controllers. The resulting controllers only use local information and outperform optimal linear droop as well as other state-of-the-art learning approaches.


\end{abstract}


%
\IEEEpeerreviewmaketitle

\section{Introduction}
\input{introduction}

\section{Model and Problem Formulation} \label{sec:model}
\input{model}

\section{Structural Properties of the Controller} \label{section:Structure Properties}
\input{structure}

\section{Learning Control Policies Using RNNs}\label{section:RNN}
\input{rnn}

\section{Simulation Studies} \label{sec:simulation}
\input{simulation}

\section{Conclusion} \label{sec:conclusion}
This paper investigates the optimal frequency control problem using reinforcement learning with stability guarantees. 
From Lyapunov stability theory, We construct the controllers to be monotonically increasing through the origin, and prove they guarantee stability for all operating points in a region. These controllers are trained using a RNN-based method that allows for efficient back propagation through time. The learned controllers are static piece-wise linear functions that do not need real-time computation and is practical for implementation. Through simulations, we show that they outperform optimal linear droop as well as purely unstructured controllers trained via reinforcement learning. In particular, controllers failing to consider stability constraints in learning may lead to unstable trajectories of the state variables, while our proposed controllers can achieve optimal performances in system frequency responses that use small control efforts. 

\bibliography{Reference.bib}

\appendices
\input{appendix}
\ifCLASSOPTIONcaptionsoff
  \newpage
\fi

\end{document}

%% file: introduction.tex
Due to the shift from conventional generation to renewable resources such as wind, solar, and storage, there has been noticeable degradation of system frequency dynamics~\cite{kroposki2017achieving}. In the near and intermediate future, both inverter-connected resources and synchronous generators would play significant roles in the grid. Therefore, the inverters still need to ``play nice'' with synchronous generators, where they need to respect the dynamics of the generators and help maintain the stability of the grid. A degradation in the frequency dynamics would increase the risk of load shedding and blackouts, which in turn limits the amount of renewable energy that can be integrated.


A widely adopted approach to use inverter-connected resources to provide primary frequency regulation is to engineer them to respond as conventional synchronous generators through frequency droop controls. 
Because of the mechanical characteristic of conventional generators, droop controls are typically linear functions of frequency deviations (with possible deadbands and saturation)~\cite{kundur1994power}. 
Inverter-connected resources can mimic this behavior by changing their active power setpoints subject to  frequency deviations~\cite{poolla2017optimal,zhang2020modeling}. 
However, as for the common performance metrics adopted in practice, including frequency deviations and control costs~\cite{zhao2014design,mallada2017optimal}, linear controllers are not optimal~\cite{Tinu20}. Since inverters are solid state electronic devices, they can implement almost arbitrary control laws by quickly adjusting their power setpoints, subject to some actuation limits~\cite{johnson2013synchronization,MPCFastFrequency}. 
Then a natural question arises: \emph{are there other control laws that still guarantee the stability of a system with synchronous generators, but have more optimal performance compared  to  linear  droop  response?}

It turns out that designing optimal controllers that respect the dynamics of power systems is not trivial. Power system dynamics are governed by nonlinear swing equations and thus even optimizing linear controllers is a difficult problem. For nonlinear controllers, they need to be parameterized in a tractable fashion for optimization. 
More importantly, the controllers need to stabilize the frequency dynamics of the grid, which introduces nonlinear constraints that are not easy to work with algebraically.

A standard approach to overcome some of the above difficulties is to work with the  linearized small signal model, where controllers can be designed to guarantee asymptotic stability~\cite{zhao2014design,mallada2017optimal}. 
However, stability becomes more crucial when state deviations are large, where the nonlinear dynamics have to be considered. When nonlinear dynamics are considered, most approaches are restricted to tuning the slopes of the linear droop controllers~\cite{zhang2020modeling}. To obtain better performances, model predictive control has also been used~\cite{MPCFastFrequency,Tinu20}, but they require robust real-time communication and computation capabilities, which is not yet available for much of the current system.


To break the unenviable position of not fully utilizing the capabilities of inverters for frequency control, a number reinforcement learning (RL) approaches have been proposed~\cite{chen2021reinforcement, yan2018data, qu2020scalable}. 
Specifically, (deep) neural networks are often used to
parameterize the controllers and RL is used to train them. A number of algorithms, including deterministic policy gradient algorithm, multi-Q-learning and actor-critic methods, have been used in frequency regulation and other control problems.

The key challenge in using RL is to guarantee that learned controllers are stabilizing, that is, frequencies in the system would reach a stable equilibrium after disturbances in the system.  
To this end, existing approaches typically use soft penalties by adding a high cost when states leave prescribed ranges~\cite{huang2019adaptive,chen2021reinforcement}.
However, these approaches are ad hoc. Stability should be treated as a hard constraint rather than through penalties, which is especially important since training can only be done on a limited number of samples while the controller should be stabilizing over a set of points in the state space. Another challenge comes from the controller training process. Generated trajectories are normally used to train the neural network controllers, but the evolution of state variables over long time horizons makes direct back-propagation inefficient. Approaches that use approximate value (or Q) function assume that the states are in a stationary probability distribution~\cite{sutton2018reinforcement}, which is generally not true during transients. Lyapunov functions have been used as constraints~\cite{ZHANG2019274}, but learning was not considered and controllers was manually tuned. 

This paper proposes a recurrent neural network (RNN)-based RL framework to solve optimal primary frequency control problem with a stability guarantee. We derive a simple algebraic condition on the nonlinear controllers that guarantee local exponential stability of the system. More precisely, using Lyapunov theory, we show that the function from the frequency deviation to the active power output implemented by a controller needs to be monotonic and through the origin at each bus. The controllers are decentralized (each only using the frequency deviation at its own bus) and the stability guarantee holds for most system parameters and topologies. 


The monotonicity of the controller is realized through a stacked-ReLU neural network which can be designed explicitly. In order to train the controllers efficiently, we design a RNN framework where the time-coupled variables in the power system form the cell component of the RNN. Simulation results show that the proposed method can learn a static nonlinear controller that performs better than traditional linear droop control. Furthermore, we show that RL without considering stability can lead to unstable controllers, whereas our approach always maintains stability. Code and data are available at https://github.com/Wenqi-Cui/RNN-RL-Frequency-Lyapunov.


In summary, the main contributions of the paper are:
\begin{enumerate}
    \item A Lyapunov function is integrated in the structural properties of controllers, guaranteeing local asymptotic stability over a large set of states. Namely, the controllers need to be monotonic functions crossing the origin.
    \item The controller is parameterized with a stacked-ReLU neural network and a RNN-based RL framework is proposed to efficiently train the controllers.
\end{enumerate}

The remaining of this paper is organized as follows.  Section~\ref{sec:model} introduces the system model and the optimal control problem. Section~\ref{section:Structure Properties} provides the main theorems governing the structure of a stabilizing controller and illustrates how it can be achieved via neural networks. Section~\ref{section:RNN} shows how they can be trained efficiently. Section~\ref{sec:simulation} shows the simulation results. Section~\ref{sec:conclusion} concludes the paper.

%% file: model.tex
\subsection{Power System Model}
Consider a $n$-bus power system that can be modelled as a connected graph $\left(\mathcal{V},\mathcal{E} \right)$. Specifically, buses are indexed by $i,j \in \mathcal{V} := \left[n\right]:=\{1,\dots, n\} $ and transmission lines are denoted by unordered pairs $\{i,j\} \in \mathcal{E} \subset \Set*{\{i,j\}}{i,j\in\mathcal{V},i\not=j}$. Let states variables be phase angle $\boldsymbol{\theta}:=\left(\theta_i, i \in \left[n\right] \right) \in \real^n$ and frequency deviation from the nominal frequency $\boldsymbol{\omega}:=\left(\omega_i, i \in \left[n\right] \right) \in \real^n$.\footnote{Throughout this paper, vectors are denoted in lower case bold and matrices are denoted in upper case bold, while scalars are unbolded.}

In this paper, we consider static local feedback controllers: bus $i$ measures its local frequency deviation $\omega_i$ and applies a time-invariant function to determine the control action $u_i$. Thus, the controller on the bus $i$ is written as $u_i(\omega_i)$. The control action changes the \emph{active power}  coming from inverter-connected resources (e.g., solar PV, electric vehicles, and storage). The detailed hardware implementation of the control laws is not the focus of this paper and is similar to existing works
(e.g., see~\cite{dominguez2012models,Xu18,HG19,jiang2021storage} and the references within).

We assume the bus voltage magnitudes are 1 per unit and the reactive power flows and injections are ignored. This is the commonly used lossless power flow model, which is suitable to primary frequency control of transmission systems with small resistances and well-regulated voltages~\cite{sauer2017power}.
Then, the  frequency dynamics  is given by the swing equation~\cite{kundur1994power}:
\begin{subequations}\label{eq:Dynamic}
\begin{align}
\dot{\theta}_i &=\omega_i\,, \label{subeq:Dynamic_delta}\\
 M_{i}\dot{\omega}_{i} &=p_{\mathrm{m},i}- D_i \omega_i-u_i({\omega}_i)- \sum_{j=1}^{n}B_{i j}\sin(\theta_i-\theta_j) \label{subeq:Dynamic_w}
\end{align}
\end{subequations}
where $\boldsymbol{M} := \diag{M_i, i \in \left[n\right]} \in \real^{n \times n}$ are the generator inertia constants, $\boldsymbol{D} := \diag{D_i, i \in \left[n\right]} \in \real^{n \times n}$ are the combined frequency response coefficients from synchronous generators and frequency sensitive load, $\boldsymbol{p}_\mathrm{m}:=\left(p_{\mathrm{m},i}, i \in \left[n\right] \right) \in \real^n$ are the net power injections, $\boldsymbol{B}:=\left[B_{ij}\right]\in \real^{n \times n}$ is the susceptance matrix 
with $B_{ij}=0, \forall  \{i,j\} \notin \mathcal{E}$, and $\boldsymbol{u}(\boldsymbol{\omega}):=\left(u_i({\omega}_i), i \in \left[n\right] \right) \in \real^n$. 

The swing equations in~\eqref{eq:Dynamic} are called the classic second-order model. As in most existing literature~\cite{dominguez2012models,Xu18,HG19,jiang2021storage}, we use them in analysis. In simulations (Section~\ref{sec:simulation}), we use $6^{nd}$-order generator model  with turbine-governing
system as well as PLL loops on the inverters for frequency measurement~\cite{chow1992toolbox,kundur1994power}.


\subsection{Optimization Problem Formulation }
As mentioned above, we would like to design control functions $u_i(\omega_i)$'s that can improve frequency deviation with a moderate control cost. Therefore, we consider two costs in the objective function of the optimal primary frequency control problem: the cost on frequency deviations
and the cost of controllers
~\cite{tabas2019optimal, mallada2017optimal,jiang2020dynamic,Tinu20}. 
For a time horizon of length $T$, a reasonable cost on frequency deviation is represented by the infinity norm of $\omega_i(t)$ over the time horizon from 0 to $T$, i.e., $\|\omega_i\|_{\infty}:=\sup_{0\leq t\leq T} |\omega_i(t)|$, which quantifies the maximum frequency deviation during the time horizon. For the cost on the control action, we use a quadratic cost defined by its two-norm $\|u_i\|_2^2:=\frac{1}{T} \int_{0}^{T} (u_i(t))^2\mathrm{d}t$. 
The optimization problem is: 
\begin{subequations}\label{eq:Continuous_Optimization}
\begin{align}
\min_{\boldsymbol{u}} & \sum_{i=1}^{n}\color{black}\left( \|\omega_i\|_{\infty}+\gamma\|u_i\|_2^2\right)\label{subeq:Continuous_Optimization_obj}\\
\mbox{s.t. } & \dot{\theta}_i =\omega_i \label{subeq:Continuous_Optimization_delta}\\
 &M_{i}\dot{\omega}_{i} =p_{\mathrm{m},i}\!-\!D_i\omega_i\!-\!u_i(\omega_i)\!-\!\!\sum_{j=1}^{n}B_{i j} \sin(\theta_i\!-\!\theta_j) \label{subeq:Continuous_Optimization_w}\\
& \underline{u}_i \leq u_i(\omega_i)\leq \overline{u}_i\label{subeq:Continuous_Optimization_bound}\\
& u_i(\omega_i) \text{ is stabilizing.}\label{subeq:Continuous_Optimization_stability}
\end{align}
\end{subequations}
Here, $\gamma$ in \eqref{subeq:Continuous_Optimization_obj} is a coefficient that trades off the cost of action with respect to frequency deviation. In a more general problem setting, distinct weights $\gamma_i$'s can be assigned to individual control actions to achieve a desirable frequency performance at an acceptable level of control action~\cite{HG19}. In practice, the power inputs from inverter-based resources are always bounded by saturation. Hence, the lower and upper bounds for the control action at bus $i$ \highlight{are included as $\underline{u}_i$ and $\overline{u}_i$, respectively, in \eqref{subeq:Continuous_Optimization_bound}}. The special case where $\underline{u}_i=\overline{u}_i=0$ can be used to characterize a bus $i$ with no controllable resources. Last but not least, we include the requirement that $u_i(\omega_i)$'s stabilize the system \eqref{eq:Dynamic} as a hard constraint in \eqref{subeq:Continuous_Optimization_stability}.  


\begin{figure}[ht]	
	\centering
	\includegraphics[scale=0.58]{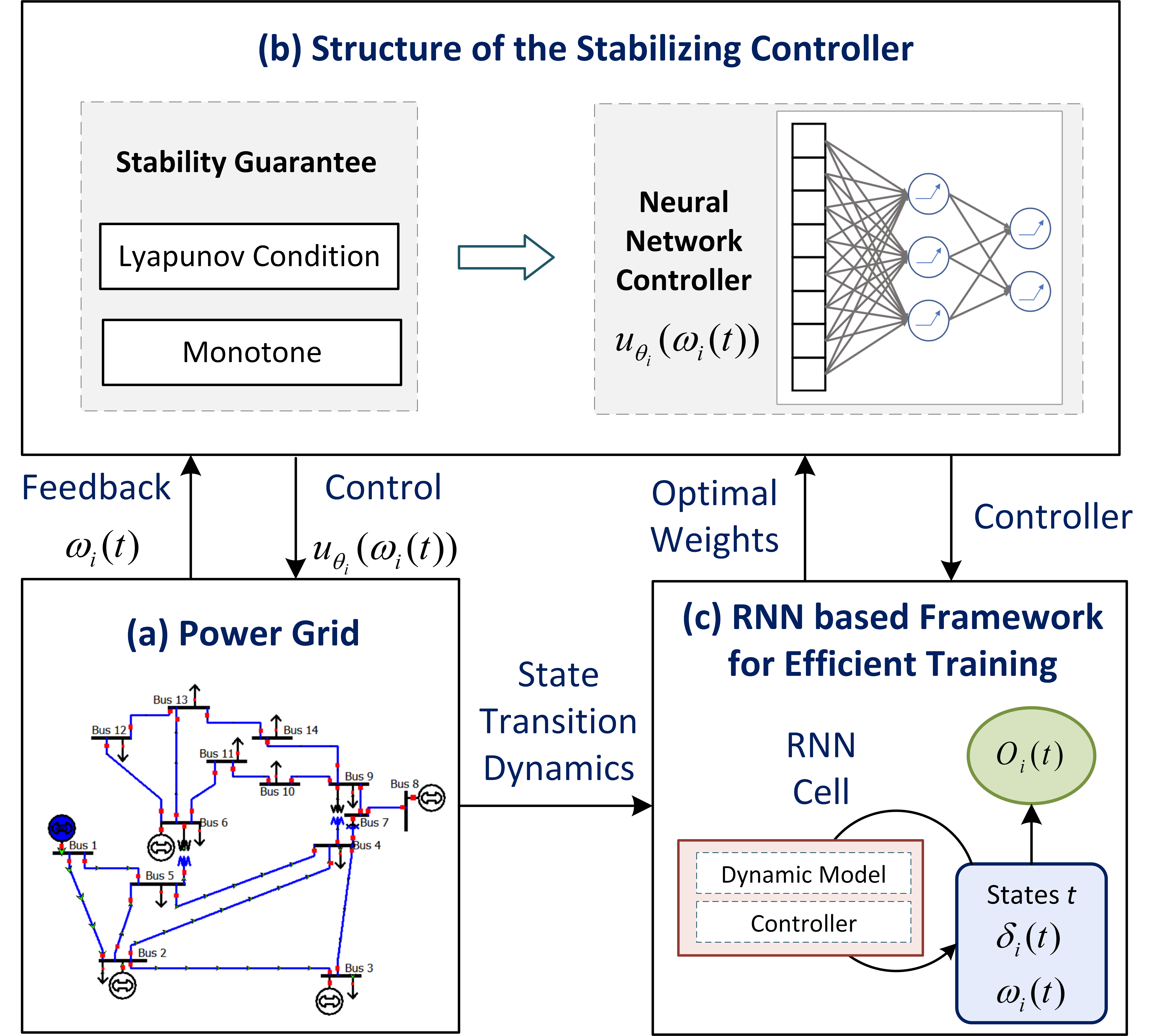}
	\caption{Reinforcement learning for the  frequency control problem}
	\label{fig:Total_Structure}
\end{figure} 

\subsection{Reinforcement Learning for Optimal Frequency Control}
In \eqref{eq:Continuous_Optimization}, we are optimizing the function $\boldsymbol{u}(\cdot)$, which is an infinite dimensional problem. To parameterize and find a good controller, reinforcement learning (RL) has emerged as an attractive alternative, where controllers are parameterized by neural networks.
Thus, we parameterize each of the controllers $u_i(\omega_i)$ as a neural network with weight $\varphi_i$, sometimes written as $u_{\varphi_i}(\omega_i)$. Then, RL trains neural networks by updating  $\varphi_i$'s to minimize the loss given by the objective function in~\eqref{subeq:Continuous_Optimization_obj}.  


The major challenge for RL comes from the hard constraint on the stability of the system. Although we can add a high penalty to the large magnitude of $\omega_i$, such a penalty does not guarantee that the stability constraints are always satisfied. In fact, learned controllers that lead to reasonably looking trajectories in training may destabilize the system during testing. To overcome this challenge, we directly use the physical model \eqref{eq:Dynamic} to derive the structure of the stabilizing controller based on Lyapunov stability theory. As illustrated in Fig.~\ref{fig:Total_Structure}(b) and discussed in Section~\ref{section:Structure Properties}, stability can be guaranteed by enforcing a structure on the controllers $u_{\varphi_i}(\omega_i)$'s. 

To use RL, we need to discretize the system dynamics in \eqref{eq:Dynamic}. The weights $\varphi_i$'s impact system behaviors across all of the time steps, which makes direct back propagation inefficient.  
Thus, we use the state transition dynamics to create a RNN framework to increase training efficiency, as illustrated in Fig.~\ref{fig:Total_Structure}(c). Details are elaborated in Section~\ref{section:RNN}.

%% file: structure.tex
To constrain the search space in \eqref{eq:Continuous_Optimization} to the set of \emph{stabilizing controllers}, we derive structural properties that the controllers should satisfy from Lyapunov stability theory. More precisely, by finding an appropriate Lyapunov function, we show that, if the output of each controller is monotonically increasing with respect to the frequency deviation, then the system has a unique equilibrium that is locally exponentially stable. In addition, we directly engineer this monotonicity feature into neural networks via properly designed weights and biases. These weights and biases are then trained to optimize the objective function in~\eqref{subeq:Continuous_Optimization_obj}.

\subsection{Uniqueness of the Equilibrium}
Since the frequency dynamics of the system in~\eqref{subeq:Dynamic_w} depends only on the phase angle differences, to characterize the equilibrium of the dynamics~\eqref{eq:Dynamic}, we make the following change of coordinates: 
$$\delta_i:=\theta_i-\dfrac{1}{n}\sum_{j=1}^{n} \theta_j\,,$$
where $\boldsymbol{\delta}:=\left(\delta_i, i \in \left[n\right] \right) \in \real^n$ can be understood as the center-of-inertia coordinates~\cite{sauer2017power,weitenberg2018robust}. Then, the system dynamics in~\eqref{eq:Dynamic} can be written as
\begin{subequations}\label{eq:equlibrium_change_cor}
    \begin{align}
\dot{\delta_i} &=\omega_i-\frac{1}{n}\sum_{j=1}^n\omega_j\,, \\
 M_{i}\dot{\omega}_{i} &=p_{\mathrm{m},i}- D_i \omega_i-u_i({\omega}_i)-\sum_{j=1}^{n}B_{i j}\sin(\delta_i-\delta_j)\,.
    \end{align}
\end{subequations}


Under an arbitrary control law $u_i(\omega_i)$, there may not exist a well-defined equilibrium point which the system will settle into. In the next lemma, we show that an unique equilibrium exists if the controllers satisfy a certain structure property.

\begin{lem}[Unique equilibrium] Suppose the function $u_i(\omega_i)$ is a monotonically increasing function of the local frequency deviation $\omega_i$. Suppose the angles satisfy $ |\delta_i-\delta_j|\in [0,\pi/2)$ for all $i$ connected to $j$. Then there exists an unique equilibrium point $(\boldsymbol{\delta}^\ast, \bm{1}\omega^\ast)$ described by
\begin{subequations}\label{eq:equli-sync}
    \begin{align}
 &0 =p_{\mathrm{m},i}- D_i \omega^\ast-u_i(\omega^\ast)-\sum_{j=1}^{n}B_{i j}\sin(\delta_i^\ast-\delta_j^\ast)\,,
 \label{eq:equli-sync-delta}\\
 &\sum_{i=1}^n p_{\mathrm{m},i}=\sum_{i=1}^nu_i(\omega^\ast)+\omega^\ast\sum_{i=1}^n D_i\,, \label{eq:equli-sync-omega}    \end{align}
\end{subequations}
if the power flow equations~\eqref{eq:equli-sync-delta} are feasible,
where $\bm{1}$ is a vector of all $1$'s with an appropriate dimension.
\end{lem}

\begin{proof} 
First of all, in steady state, \eqref{eq:equlibrium_change_cor} yields
\begin{subequations}\label{eq:equli}
    \begin{align}
 0 &=\omega_i^\ast-\frac{1}{n}\sum_{j=1}^n\omega_i^\ast\,, \label{eq:omega-i-eq}\\
 0 &=p_{\mathrm{m},i}- D_i \omega_i^\ast-u_i(\omega_i^\ast)-\sum_{j=1}^{n}B_{i j}\sin(\delta_i^\ast-\delta_j^\ast)\,.\label{eq:delta-i-eq}
    \end{align}
\end{subequations}
Clearly, \eqref{eq:omega-i-eq} implies that the frequency deviation at each bus synchronizes to the same solution that $\omega_i^\ast = \omega^\ast$, and we have the desired equations in \eqref{eq:equli-sync-delta}. Since the system is lossless and $B_{ij}=B_{ji}$, the net power flow, $\sum_{i=1}^{n}\sum_{j=1}^{n}B_{i j}\sin(\delta_i^\ast-\delta_j^\ast)$, is zero. Using this fact and by summing~\eqref{eq:equli-sync-delta}, we get \eqref{eq:equli-sync-omega}.  

Next, we show the uniqueness of $\omega^\ast$ by contradiction. Suppose that both $\omega^\ast$ and $\omega^\star$ satisfy \eqref{eq:equli-sync-omega}, where $\omega^\ast\neq\omega^\star$. Then, $$\sum_{i=1}^n u_i(\omega^\ast)+\omega^\ast\sum_{i=1}^n D_i=\sum_{i=1}^n u_i(\omega^\star)+\omega^\star\sum_{i=1}^n D_i\,,$$
which yields
\begin{equation}\label{eq:uni-omega-pro}
\sum_{i=1}^n \dfrac{u_i(\omega^\ast)-u_i(\omega^\star)}{\omega^\ast-\omega^\star}=-\sum_{i=1}^n D_i<0\,.    
\end{equation}
However, if $u_i(\omega_i)$ is monotonically increasing, the left hand side of the equality in \eqref{eq:uni-omega-pro} must be nonnegative, which is a contradiction. The uniqueness of $\boldsymbol{\delta}^\ast$ follows from the same argument as in~\cite[Lemma~1]{weitenberg2018exponential}.
\end{proof}

Note that the angles  $\boldsymbol{\delta}$ are constrained to be in the region denoted by $\Theta :=\left \{\boldsymbol{\delta}| |\delta_i-\delta_j|\in [0,\pi/2), \forall \{i,j\}\in \mathcal{E} \right \}$, which is sufficiently large to include almost all practical scenarios and is a common assumption in literature~\cite{weitenberg2018robust, sauer2017power}.

\subsection{Lyapunov Stability Analysis}
In this subsection, we further show that the equilibrium point $(\boldsymbol{\delta}^\ast, \boldsymbol{\omega}^\ast)$ described by \eqref{eq:equli-sync} is locally exponentially stable if the controllers are monotone. The next
theorem is the main result of the paper.
\begin{thm}[Local exponential stability]\label{theorem:Exponential_stable}
If the control output $u_i(\omega_i)$ is a monotonically increasing function of the local frequency deviation $\omega_i$, then the equilibrium point $(\boldsymbol{\delta}^\ast, \bm{1}\omega^\ast)$ described by \eqref{eq:equli-sync} is locally exponentially stable. In particular, the region of attraction include the set $\mathcal{D}:=\! \Set*{(\boldsymbol{\delta}, \boldsymbol{\omega})\in \real^n \times \real^n}{|\delta_i-\delta_j|\in [0,\pi/2) \mbox{ for $i,j$ connected}}$.
\end{thm}
The qualifier ``local'' in Theorem~\ref{theorem:Exponential_stable} is necessary since we need to assume that the trajectories start within the region of attraction. We note that this is far less restrictive than standard local convergence results in nonlinear systems, where the region of attraction is confined to be close to the equilibrium point~\cite{sastry2013nonlinear}. The region of attraction in Theorem~\ref{theorem:Exponential_stable} is quite large and include most operating points of interest. 

Theorem~\ref{theorem:Exponential_stable} gives structural properties\footnote{These are sometimes called extended class $\kappa$ functions} for controllers that guarantee exponential stability that does not depend on system parameter and topologies.  
Therefore, the optimal performance comes from training on a particular system, but the stability guarantees do not. This robustness to uncertainties is a key advantage of constraining the structure of networks compared to purely model-free RL approaches. The design of neural networks is given in the next section (Section~\ref{sec:design}) and the rest of this section outlines the proof of Theorem~\ref{theorem:Exponential_stable}. 

From Lyapunov stability theory, if there exists a Lyapunov function $V(\boldsymbol{\delta}, \boldsymbol{\omega})$ such that $\dot{V}(\boldsymbol{\delta}, \boldsymbol{\omega})\leq -cV(\boldsymbol{\delta}, \boldsymbol{\omega})$ for a constant $c>0$, then the system is exponentially stable~\cite{sastry2013nonlinear}. Therefore, we prove Theorem~\ref{theorem:Exponential_stable} by constructing a qualified Lyapunov function  and showing that such a constant $c$ exist.
Inspired by~\cite{weitenberg2018robust}, we consider the following Lyapunov function candidate:
\begin{align}\label{eq: Lyapunov}
V(\boldsymbol{\delta}, \boldsymbol{\omega})=&\ \frac{1}{2}\sum_{i=1}^{n} M_i(\omega_{i}-\omega^\ast)^{2}+W_\mathrm{p}(\boldsymbol{\delta})+\epsilon W_\mathrm{c}(\boldsymbol{\delta}, \boldsymbol{\omega})
\end{align}
with
\begin{subequations}
\begin{align}
    W_\mathrm{p}(\boldsymbol{\delta}) := 
&-\frac{1}{2}\sum_{i=1}^{n} \sum_{j=1}^{n} B_{ij}\left( \cos(\delta_{ij}) -\cos(\delta_{ij}^*)\right)\nonumber\\
&-\sum_{i=1}^{n}\sum_{j=1}^{n} B_{ij} \sin(\delta_{ij}^*)( \delta_{i}-\delta_{i}^*)\,,\label{eq:Wp}\\
W_\mathrm{c}(\boldsymbol{\delta}, \boldsymbol{\omega}):=&\sum_{i=1}^{n} \sum_{j=1}^{n}B_{ij}\left(  \sin(\delta_{ij})-\sin(\delta_{ij}^*)\right)M_i(\omega_i-\omega^\ast)\,,\nonumber
\end{align}
\end{subequations}
where $\delta_{ij}:=\delta_{i}-\delta_{j}$ and $\epsilon>0$ is a tunable parameter that should be set small enough. The physical intuition for the  Lyapunov function can be found in~\cite{arapostathis1982global, weitenberg2018robust}. 
Strictly speaking, this function is not a ``true'' Lyapunov function since it is not bounded below. 
The following lemma proves that $V(\boldsymbol{\delta}, \boldsymbol{\omega})$ is a well-defined Lyapunov function on the domain $\mathcal{D}$, which suffices to show that trajectories starting in $\mathcal{D}$ converge to the equilibrium. Then Lemma~\ref{lemma: dot_V} derives the time derivative  $\dot{V}(\boldsymbol{\delta}, \boldsymbol{\omega})$ and Lemma~\ref{lemma:Exponential_stable} shows there exists a constant $c>0$ such that $\dot{V}(\boldsymbol{\delta}, \boldsymbol{\omega})\leq -cV(\boldsymbol{\delta}, \boldsymbol{\omega})$.


\begin{lem}[Bounds on Lyapunov function]\label{lemma:Pos_Lyapunov}
$\forall (\boldsymbol{\delta}, \boldsymbol{\omega})\!\in\! \mathcal{D}$, the Lyapunov function $V(\boldsymbol{\delta}, \boldsymbol{\omega})$ in~\eqref{eq: Lyapunov} satisfies
$$
\begin{aligned}
&V(\boldsymbol{\delta}, \boldsymbol{\omega}) \geq\alpha_{1}(\|\bm{\delta}-\bm{\delta}^*\|_2^{2}+\|\bm{\omega}-\bm{\omega}^*\|_2^{2})\,,\\ 
& V(\boldsymbol{\delta}, \boldsymbol{\omega}) \leq \alpha_{2}(\|\bm{\delta}-\bm{\delta}^*\|_2^{2}+\|\bm{\omega}-\bm{\omega}^*\|_2^{2})\,,
\end{aligned}
$$
for some constants $\alpha_{1}>0$ and $\alpha_{2}>0$.
\end{lem}

The proof is given in Appendix~\ref{appendices:lemma_Pos_Lyapunov}. It follows directly from Lemma~\ref{lemma:Pos_Lyapunov} that $V(\boldsymbol{\delta}^\ast, \boldsymbol{\omega}^\ast)=0$ and $V(\boldsymbol{\delta}, \boldsymbol{\omega})>0, \forall(\boldsymbol{\delta}, \boldsymbol{\omega})\in \mathcal{D}\setminus{(\boldsymbol{\delta}^\ast, \boldsymbol{\omega}^\ast)}$. To show $V(\boldsymbol{\delta}, \boldsymbol{\omega})$ is a Lyapunov function on $\mathcal{D}$, we need to show $\dot{V}(\boldsymbol{\delta}, \boldsymbol{\omega})$ decreases in $\mathcal{D}$.

\begin{lem}[Time derivative]\label{lemma: dot_V}
The time derivative of $V(\boldsymbol{\delta}, \boldsymbol{\omega})$ defined in~\eqref{eq: Lyapunov} is given by 
\begin{align}\label{eq:Vdot}
\dot{V}(\boldsymbol{\delta}, \boldsymbol{\omega}) =&-\!\begin{bmatrix}\boldsymbol{p}_\mathrm{e}(\boldsymbol{\delta})\!-\!\boldsymbol{p}_\mathrm{e}(\boldsymbol{\delta}^\ast)\\\boldsymbol{\omega}\!-\!\boldsymbol{\omega}^\ast\end{bmatrix}^T\!\!\boldsymbol{Q}(\boldsymbol{\delta})\begin{bmatrix}\boldsymbol{p}_\mathrm{e}(\boldsymbol{\delta})\!-\!\boldsymbol{p}_\mathrm{e}(\boldsymbol{\delta}^\ast)\\\boldsymbol{\omega}\!-\!\boldsymbol{\omega}^\ast\end{bmatrix}\\
&-\left[\boldsymbol{\omega}\!-\!\boldsymbol{\omega}^\ast \!\!+\! \epsilon\left(\boldsymbol{p}_\mathrm{e}(\boldsymbol{\delta})\!-\!\boldsymbol{p}_\mathrm{e}(\boldsymbol{\delta}^\ast)\right)\right]^T\left(\boldsymbol{u}(\boldsymbol{\omega})-\boldsymbol{u}(\boldsymbol{\omega}^\ast)\right)\nonumber
\end{align}
with
\begin{equation}
\boldsymbol{Q}(\boldsymbol{\delta}):=
    \begin{bmatrix}
\epsilon\boldsymbol{I}  & \dfrac{\epsilon}{2}\boldsymbol{D} \\  \dfrac{\epsilon}{2}\boldsymbol{D}  & \boldsymbol{D}-\dfrac{\epsilon}{2}(\boldsymbol{H}(\boldsymbol{\delta})\boldsymbol{M}+\boldsymbol{M}\boldsymbol{H}(\boldsymbol{\delta}))
\end{bmatrix},
\end{equation}
which is positive definite for $\epsilon$  small enough,
 $\boldsymbol{p}_\mathrm{e}(\boldsymbol{\delta}) := \left(p_{\mathrm{e},i}(\boldsymbol{\delta}):=\sum_{j=1}^{n} B_{ij} \sin(\delta_{ij}), i \in \left[n\right] \right) \in \real^n$ and $\boldsymbol{H}(\boldsymbol{\delta})=\nabla \boldsymbol{p}_\mathrm{e}(\boldsymbol{\delta}):=\left[ H_{ij}\right]\in \real^{n \times n}$ such that 
\begin{align}\label{eq:Lij}
    H_{ij}:=
    \begin{dcases}
    -B_{ij}\cos(\delta_{ij}) & \text{if}\ i\neq j\\
    \sum_{j^\prime=1,j^\prime\neq i}^n B_{ij\prime}\cos(\delta_{ij\prime})& \text{if}\ i=j
    \end{dcases}\,,\qquad\forall i,j \in \left[n\right]\,.
\end{align}
\end{lem}

The proof is given in Appendix~\ref{appendices:lemma_dot_V}.
The cross term $\left[\boldsymbol{\omega}-\boldsymbol{\omega}^\ast + \epsilon\left(\boldsymbol{p}_\mathrm{e}(\boldsymbol{\delta})-\boldsymbol{p}_\mathrm{e}(\boldsymbol{\delta}^\ast)\right)\right]^T\left(\boldsymbol{u}(\boldsymbol{\omega})-\boldsymbol{u}(\boldsymbol{\omega}^\ast)\right)$ generally complicates the analysis of $\dot{V}(\boldsymbol{\delta}, \boldsymbol{\omega})$. But when $u_i(\omega_i)$ is monotonically increasing with respect to $\omega_i$,  $\left(u_i(\omega_i)-u_i(\omega^*)\right)$ is the same sign with $\left(\omega_i-\omega^*\right)$ and leads to nonnegative cross terms for small $\epsilon$, implying that  $\dot{V}(\boldsymbol{\delta}, \boldsymbol{\omega})<0$, $\forall(\boldsymbol{\delta}, \boldsymbol{\omega})\in \mathcal{D}\setminus{(\boldsymbol{\delta}^\ast, \boldsymbol{\omega}^\ast)}$ and thus the system is locally asymptotically stable at the equilibrium point $(\boldsymbol{\delta}^\ast, \boldsymbol{\omega}^\ast)$. In the next lemma, we further show local exponential stability of the equilibrium.  


\begin{lem}[Bounds on the time derivative]\label{lemma:Exponential_stable}
If $u_i(\omega_i)$ is monotonically increasing with respect to $\omega_i$, then there exists a constant $c>0$ such that $\dot{V}(\boldsymbol{\delta}, \boldsymbol{\omega})\leq -cV(\boldsymbol{\delta}, \boldsymbol{\omega})$.
\end{lem}

\begin{proof}

First, we show that the cross term related to $u_i(\omega_i)$ is nonnegative for sufficiently small $\epsilon$.
Define
\begin{align*}
    k_{i}(\omega_i):=
    \begin{dcases}
    \dfrac{u_i(\omega_i)-u_i(\omega_i^\ast)}{\omega_i-\omega_i^\ast} & \text{if}\ \omega_i \neq\omega_i^\ast\\
    0& \text{if}\ \omega_i=\omega_i^\ast
    \end{dcases}\,,\quad\forall i \in \left[n\right]\,.
\end{align*}
Then, $\boldsymbol{K}(\boldsymbol{\omega}) := \diag{k_{i}(\omega_i), i \in \mathcal{V}} \in \real^{n \times n}\succeq0$ if $u_i(\omega_i)$ is monotonically increasing with respect to $\omega_i$. Hence, 
\begin{align}\label{eq:bound-cross}
&\left[\boldsymbol{\omega}-\boldsymbol{\omega}^\ast + \epsilon\left(\boldsymbol{p}_\mathrm{e}(\boldsymbol{\delta})-\boldsymbol{p}_\mathrm{e}(\boldsymbol{\delta}^\ast)\right)\right]^T\left(\boldsymbol{u}(\boldsymbol{\omega})-\boldsymbol{u}(\boldsymbol{\omega}^\ast)\right)\nonumber\\
=&\left(\boldsymbol{\omega}-\boldsymbol{\omega}^\ast\right)^T\boldsymbol{K}(\boldsymbol{\omega})\left(\boldsymbol{\omega}-\boldsymbol{\omega}^\ast\right)\nonumber\\ &+\epsilon\left(\boldsymbol{p}_\mathrm{e}(\boldsymbol{\delta})-\boldsymbol{p}_\mathrm{e}(\boldsymbol{\delta}^\ast)\right)^T\boldsymbol{K}(\boldsymbol{\omega})\left(\boldsymbol{\omega}-\boldsymbol{\omega}^\ast\right) \geq 0 \nonumber
\end{align}
for small enough $\epsilon$.

Then, $\dot{V}(\boldsymbol{\delta}, \boldsymbol{\omega})$ can be by bounded by the quadratic term related to $\boldsymbol{Q}(\boldsymbol{\delta})$ in \eqref{eq:Vdot} as follows:
\begin{align}
    &\dot{V}(\boldsymbol{\delta}, \boldsymbol{\omega})\\
    &\leq
    -\begin{bmatrix}\boldsymbol{p}_\mathrm{e}(\boldsymbol{\delta})\!-\!\boldsymbol{p}_\mathrm{e}(\boldsymbol{\delta}^\ast)\nonumber\\\boldsymbol{\omega}\!-\!\boldsymbol{\omega}^\ast\end{bmatrix}^T\!\!\boldsymbol{Q}(\boldsymbol{\delta})\begin{bmatrix}\boldsymbol{p}_\mathrm{e}(\boldsymbol{\delta})\!-\!\boldsymbol{p}_\mathrm{e}(\boldsymbol{\delta}^\ast)\nonumber\\\boldsymbol{\omega}\!-\!\boldsymbol{\omega}^\ast\end{bmatrix}\nonumber\\
    &\stackrel{(\mathrm{a})}{\leq} - \lambda_\mathrm{min}(\boldsymbol{Q}(\boldsymbol{\delta}))\left(\|\boldsymbol{p}_\mathrm{e}(\boldsymbol{\delta})\!-\!\boldsymbol{p}_\mathrm{e}(\boldsymbol{\delta}^\ast)\|^2_2+\|\boldsymbol{\omega}\!-\!\boldsymbol{\omega}^\ast\|^2_2\right)
    \nonumber\\
    &\stackrel{(\mathrm{b})}{\leq} - \lambda_\mathrm{min}(\boldsymbol{Q}(\boldsymbol{\delta}))\left(\gamma_1\|\boldsymbol{\delta}\!-\!\boldsymbol{\delta}^\ast\|^2_2+\|\boldsymbol{\omega}\!-\!\boldsymbol{\omega}^\ast\|^2_2\right)
    \nonumber\\
    &  \leq -\lambda_\mathrm{min}(\boldsymbol{Q}(\boldsymbol{\delta}))\min(1, \gamma_1)\left(\|\boldsymbol{\delta}\!-\!\boldsymbol{\delta}^\ast\|^2_2+\|\boldsymbol{\omega}\!-\!\boldsymbol{\omega}^\ast\|^2_2\right)
    \nonumber\\
    &\stackrel{(\mathrm{c})}{\leq} - \lambda_\mathrm{min}(\boldsymbol{Q}(\boldsymbol{\delta}))\min(1, \gamma_1)\dfrac{1}{\alpha_2}V(\boldsymbol{\delta}, \boldsymbol{\omega}) \nonumber
    \\
    & \leq- c V(\boldsymbol{\delta}, \boldsymbol{\omega})\label{eq:bound-Q}
\end{align}
with $$c:=\left(\min_{\boldsymbol{\delta}:|\delta_i-\delta_j|\in [0,\pi/2), \forall \{i,j\}\in \mathcal{E}} \!\lambda_\mathrm{min}(\boldsymbol{Q}(\boldsymbol{\delta}))\right)\dfrac{\min(1, \gamma_1)}{\alpha_2}>0,$$
where (a) is given by the Rayleigh-Ritz theorem, (b) is  by~\cite[Lemma~4]{weitenberg2018exponential} with $
\gamma_{1}:=\min _{\tilde{\delta} \in \Theta} \lambda_{2}(\bm{H}(\bm{\tilde{\delta}}))^{2}$, and (c) follows from  Lemma~\ref{lemma:Pos_Lyapunov}.
\end{proof}

\subsection{Design of Neural Network Controllers} \label{sec:design}
In this paper, we parametrize the controllers $u_{\varphi_i}(\omega_i)$ by a single hidden layer neural network. 
We assume that the processes such as automatic generation
control (AGC) adjust the power setpoint of generators to make the net power injection around zero, i.e., $\sum_{i=1}^{n}p_{\mathrm{m},i}=0$. For controllers $u_i(\omega_i)$'s that provide primary frequency response, we set $u_i(0)=0$ so the controllers take no action when there is no frequency deviation. 
By Theorem~\ref{theorem:Exponential_stable}, we design
the neural networks to have the following structures \highlight{such that  the controller will be locally exponentially stabilizing}:
\begin{enumerate}
    \item $u_{\varphi_i}(\omega_i)$ is monotonically increasing;
    \item $u_{\varphi_i}(\omega_i)=0$ for $\omega_i=0$;
    \item  $\underline{u}_i \leq u_{\varphi_i}(\omega_i)\leq \overline{u}_i$ (saturation constraints). 
\end{enumerate}

The first two requirements are equivalent to designing a monotonic increasing function through the origin. This is constructed by decomposing the function into positive and negative parts as $f_i(\omega_i)=f_i^+(\omega_i)+f_i^-(\omega_i)$, where $f_i^+(\omega_i)$ is monotonic increasing for $\omega_i>0$ and zero when $\omega_i\leq0$; $f_i^-(\omega_i)$ is monotonic increasing for $\omega_i<0$ and zero when $\omega_i\geq0$. The saturation constraints can be satisfied by hard thresholding the output of the neural network.

The function $f_i^+(\omega_i)$ and $f_i^-(\omega_i)$ are constructed using a single-layer neural network designed by stacking the ReLU function  $\sigma(x)=\max(x,0)$.
Let $m$ be the number of hidden units. For $f_i^+(\omega_i)$, let  $\bm{q}_i=[\begin{matrix}
 q_i^{1}& q_i^{2} & \cdots&q_i^{m}
\end{matrix}]$ be the weight vector of bus $i$; $\bm{b}_i=[\begin{matrix} b_i^{1}& b_i^{2} & \cdots &b_i^{m}\end{matrix}]^\intercal$ be the corresponding bias vector. For $f_i^-(\omega_i)$, let $\bm{z}_i=[\begin{matrix} z_i^{1}&z_i^{2}&\cdots&z_i^{m}\end{matrix}]$ be the weights vector and $\bm{c}_i=[\begin{matrix} c_i^{1}& c_i^{2} & \cdots & c_i^{m}\end{matrix}]^\intercal$ be the bias vector. Denote $\mathbf{1}\in \mathbb{R}^m$ as the all $1$'s column vector. The detailed construction of $f_i^+(\omega_i)$ and $f_i^-(\omega_i)$ is given in Lemma~\ref{lemma:u_f+}. 
\begin{lem}\label{lemma:u_f+}
Let $\sigma(x)=\max(x,0)$ be the ReLU function.  The stacked ReLU function constructed by \eqref{eq:f+} is monotonic increasing for $\omega_i>0$ and zero when $\omega_i\leq 0$.
\begin{subequations}\label{eq:f+}
\begin{align}
     &f_i^+(\omega_i)=\bm{q}_i\sigma(\mathbf{1}\omega_i+\bm{b}_i)\\
     \mbox{where }
     &\sum_{j=1}^{l} q_i^{j}\geq 0, \quad \forall l=1,2,\cdots,m
     \label{eq:f+2}\\
     & b_i^{1}=0, b_i^{l}\leq b_i^{(l-1)}, \quad \forall l=2,3,\cdots,m\label{eq:f+3}
\end{align}
\end{subequations}
 The stacked ReLU function constructed by \eqref{eq:f-} is monotonic increasing for $\omega_i< 0$ and zero when $\omega_i\geq 0$.
\begin{subequations}\label{eq:f-}
\begin{align}
     &f_i^-(\omega_i)=\bm{z}_i\sigma(-\mathbf{1}\omega_i+\bm{c}_i)\\
     \mbox{where } 
     &\sum_{j=1}^{l} z_i^{j}\leq 0, \quad \forall l=1,2,\cdots,m\label{eq:f-2}\\
     & c_i^{1}=0, c_i^{l}\leq c_i^{(l-1)}, \quad \forall l=2,3,\cdots,m\label{eq:f-3}
\end{align}
\end{subequations}
\end{lem}

\begin{proof}
 Note that the ReLU function $\sigma(x)$ is linear with $x$ when activated $(x>0)$ and equals to zero when deactivated $(x \leq 0)$, we construct the monotonic increasing function $f_i^+(\omega_i)$ by stacking the function $g_i^l(\omega_i)=q_i^{l}\sigma(\omega_i+b_i^{l})$, 
 as illustrated by Fig.~\ref{fig:Mono_Structure}. 
 Since
$b_i^{1}=0$ and $b_i^{l}\leq b_i^{(l-1)}, \forall 1\leq l \leq m$, $g_i^l(\omega_i)$  is activated in sequence from $g_i^1(\omega_i)$ to $g_i^m(\omega_i)$ with the increase of $\omega_i$. In this way, the stacked function is a piece-wise linear function and the slope for each piece is $\sum_{j=1}^{l} q_i^j$. 
Monotonic property can be satisfied as long as the slope of all the pieces are positive, i.e., $\sum_{j=1}^{l} q_i^l\geq 0, \forall 1\leq l \leq m$. Similarly, $f_i^-(\omega_i)$ also construct by ReLU function activated for negative $w_i$ in sequence corresponding to $c_i^l$ for $l=1,\cdots,m$. $\sum_{j=1}^{l} z_i^{j}\leq 0$ means that all the slope of the piece-wise linear function is positive and therefore guarantees monotonicity.
\begin{figure}[t]	
	\centering
	\includegraphics[scale=0.65]{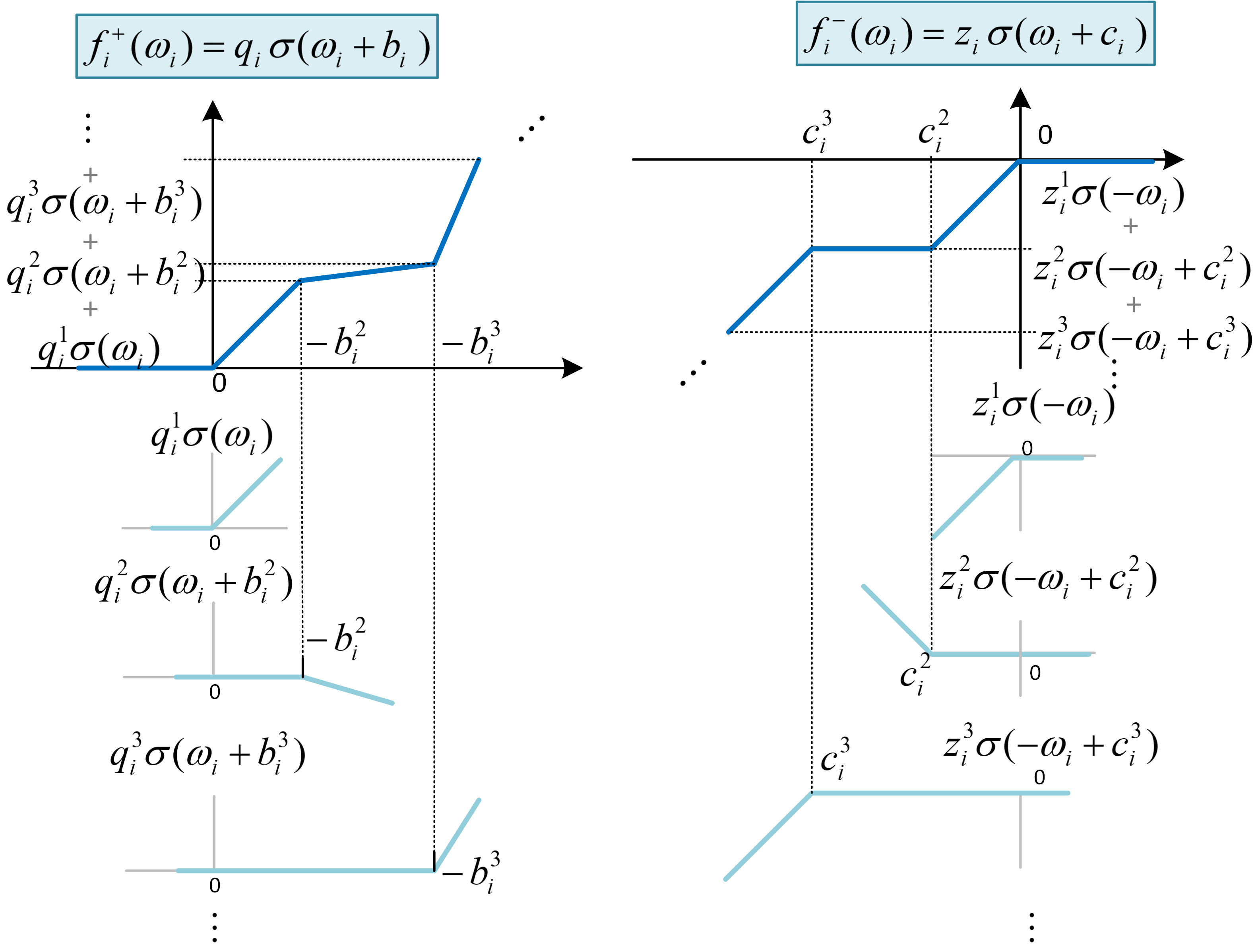}
	\caption{Stacked ReLU neural network to formulate a monotonic increasing function through the origin}
	\label{fig:Mono_Structure}
\end{figure}
\end{proof}
Note that there still exists inequality constraints in \eqref{eq:f+} and \eqref{eq:f-}, which makes the training of the neural networks cumbersome.  We can reformulate the weights to get an equivalent representation that is easier to deal with in training. Define the non-negative vectors $\hat{\bm{q}}_i=\left[\begin{matrix} \hat{q}_i^{1}&\cdots&\hat{q}_i^{m}\end{matrix}\right]$ and $\hat{\bm{b}}_i=\left[\begin{matrix}\hat{b}_i^{1}&\cdots&\hat{b}_i^{m}\end{matrix}\right]^\intercal$. Then, \eqref{eq:f+2} is satisfied if $q_i^{1}=\hat{q}_i^{1}$, $q_i^{l}=\hat{q}_i^{l}-\hat{q}_i^{(l-1)}$ for  $l=2,\cdots,m$. \eqref{eq:f+3} is satisfied if $b_i^{1}=0$, $ b_i^{l}=-\sum_{j=2}^{l} \hat{b}_i^{j}$ for  $l=2,\cdots,m$. 
Similarly, define $\hat{\bm{z}}_i=\left[\begin{matrix}\hat{z}_i^{1}&\cdots&\hat{z}_i^{m}\end{matrix}\right]\geq 0$  and $\hat{\bm{c}}_i=\left[\begin{matrix}\hat{c}_i^{1}&\cdots&\hat{c}_i^{m}\end{matrix}\right]^\intercal\geq 0$. Then, \eqref{eq:f-2} is satisfied if $z_i^{1}=-\hat{z}_i^{1}$, $ z_i^{l}=-\hat{z}_i^{l}+\hat{z}_i^{(l-1)}$ for  $l=2,\cdots,m$. \eqref{eq:f-3} is satisfied if $c_i^{1}=0$, $ c_i^{l}=-\sum_{j=2}^{l} \hat{c}_i^{l}$ for  $l=2,\cdots,m$. If the dead-band of the frequency deviation within the range $[-d,d]$ is required, it can be easily satisfied by setting $b_i^2=-d, q_i^1=0$ and  $c_i^2=-d, z_i^1=0$ in \eqref{eq:f+} and \eqref{eq:f-}.\footnote{A deadband is often enforced for generator droop control to reduce mechanical stress. For inverters, we do not set mandatory dead-bands.}

The next Theorem states the converse of Lemma~\ref{lemma:u_f+}, that is, the constructions in \eqref{eq:f+} and \eqref{eq:f-} suffice to approximate all functions of interest. 
\begin{thm}\label{theorem:Universial Approximation}
Let $r(x)$ be any continuous, Lipschitz and bounded monotonic function through the origin with bounded derivatives, mapping compact set $\mathbb{X}$ to  $ \mathbb{R}$.  Then there exists a function $f(x)=f^+(x)+f^-(x)$ constructed by \eqref{eq:f+} and \eqref{eq:f-} such that, for any $\epsilon$ and any $x \in \mathbb{X}$, $\left|r(x)-f(x)\right|<\epsilon$. 
\end{thm}

The proof is given in Appendix~\ref{appendices:stacked-ReLU}. 
Note that $f(x)$ is a single-layer neural network. When approximating an arbitrary function, the number of neurons and the height will depend on $\epsilon$.
Since the controller in this paper is bounded, the stacked-ReLU neural network with limited number of neurons is sufficient for parameterization. The last step is to bound the output of the neural networks, which can be done easily using ReLU activation functions. 
\begin{lem}\label{lemma:u_f+-}
 The neural network controller $u_i(\omega_i)$ given below is a monotonic increasing function through the origin and bounded in $[\underline{u}_i, \overline{u}_i]$ for all $i=1,\cdots,N$:
\begin{equation}\label{eq:u_f+-}
\begin{split}
    u_i(\omega_i)=&\overline{u}_i-\sigma(\overline{u}_i -f_i^+(\omega_i)-f_i^-(\omega_i))\\
    &+\sigma(\underline{u}_i -f_i^+(\omega_i)-f_i^-(\omega_i))
\end{split}
\end{equation}
\end{lem}
The proof of this lemma is by inspection.

%% file: rnn.tex
The structure of the controllers are decided by the constructions in \eqref{eq:f+}, \eqref{eq:f-} and \eqref{eq:u_f+-}. In this section we develop a RNN based RL algorithm to learn their weights and biases. 
\subsection{Discretize Time System}
To learn the controller and simulate the trajectories of the system, we discretize the dynamics \eqref{eq:Dynamic} with step size $\Delta t$. We use $k$ and $K$ to represent the discrete time and the total number of stages, respectively.  The states $(\theta_i,w_i)$ at bus $i$ evolves along the trajectory are represented as $\bm{\theta_i}=(\theta_i(0),\theta_i(1),\cdots,\theta_i(K) )$ and $\bm{\omega_i}= (\omega_i(0),\omega_i(1),\cdots,\omega_i(K))$ over K stages,  with the control sequence $\bm{u_{\varphi_i}}=(u_{\varphi_i}(\omega_i(0)),\cdots,u_{\varphi_i}(\omega_i(K)))$. The infinity norm of the sequence of $\omega_i(k)$ is then defined by $||\bm{\omega_i}||_{\infty}=\max_{k=0,\cdots,K} |w_i(k)|$. The cost on controller is the quadratic function of action 
defined by its two-norm $||\bm{u_{\varphi_i}}||_2^2=\frac{1}{K}\sum_{k=1}^{K} (u_{\varphi_i}(k))^2$ . The optimization problem is
\begin{subequations}\label{eq:Discrete_optimization}
\begin{align}
\min_{\bm{\varphi}}&\quad  \sum_{i=1}^{n} \color{black}\left(||\mathbf{\bm{\omega_i}}||_{\infty}+\gamma||\bm{u_{\varphi_i}}||_2^2\right)\label{subeq:Discrete_optimization_Obj}\\
 \mbox{s.t. } & \theta_i(k)= \theta_i(k-1)+ \omega_i(k-1)\Delta t \label{subeq:Discrete_optimization_delta}\\
\begin{split}
&\omega_i(k) =-\frac{\Delta t}{M_{i}}\sum_{j=1}^{|\mathcal{B}|}B_{i j}\sin(\theta_{ij}(k-1))+\frac{\Delta t}{M_{i}}p_{m,i}\\
&\quad+\left(1-\frac{D_i\Delta t}{M_{i}}\right)\omega_i(k-1)-\frac{\Delta t}{M_{i}}u_{\varphi_i}(\omega_i(k-1))
\end{split}\label{subeq:Discrete_optimization_w} \\
& \underline{u}_i \leq u_{\varphi_i}(\omega_i(k))\leq \overline{u}_i \label{subeq:Discrete_optimization_bound}\\
& \omega_i(k) u_{\varphi_i}(\omega_i(k))\geq 0 \label{subeq:Discrete_optimization_Lyapunov}\\
&  u_{\varphi_i} (\cdot) \mbox{ is increasing}
\label{subeq:Discrete_optimization_Mono}
\end{align}
\end{subequations}
and all equations hold for $i=1,\dots,n$. The constraints \eqref{subeq:Discrete_optimization_Lyapunov} and \eqref{subeq:Discrete_optimization_Mono} guarantee exponentially stability. 

Note that the optimization variable $\bm{\varphi}$ exists in all the time steps in \eqref{eq:Discrete_optimization}. A straightforward gradient-based training approach is challenging since we need to calculate the gradient all the way to the first time step for all time steps $k=0,\cdots,K$. 
To mitigate this challenge, we propose a RNN-based framework that integrates the state transition dynamics \eqref{subeq:Discrete_optimization_delta} and \eqref{subeq:Discrete_optimization_w}.  
This way, the gradient of the optimization objective with respect to $\bm{\varphi}$ can be computed efficiently through back-propagation.

\subsection{RNN for control}
RNN is a class of artificial neural networks where connections between nodes form a directed graph along a temporal sequence. This allows it to exhibit temporal dynamic behavior. 
 By defining the cell state as the time-coupled states $\theta_i$ and $\omega_i$, the state transition dynamics of the power system is integrated as  illustrated in Fig.~\ref{fig:RNN_Structure}\\
\begin{figure}[H]	
	\centering
	\includegraphics[scale=0.48]{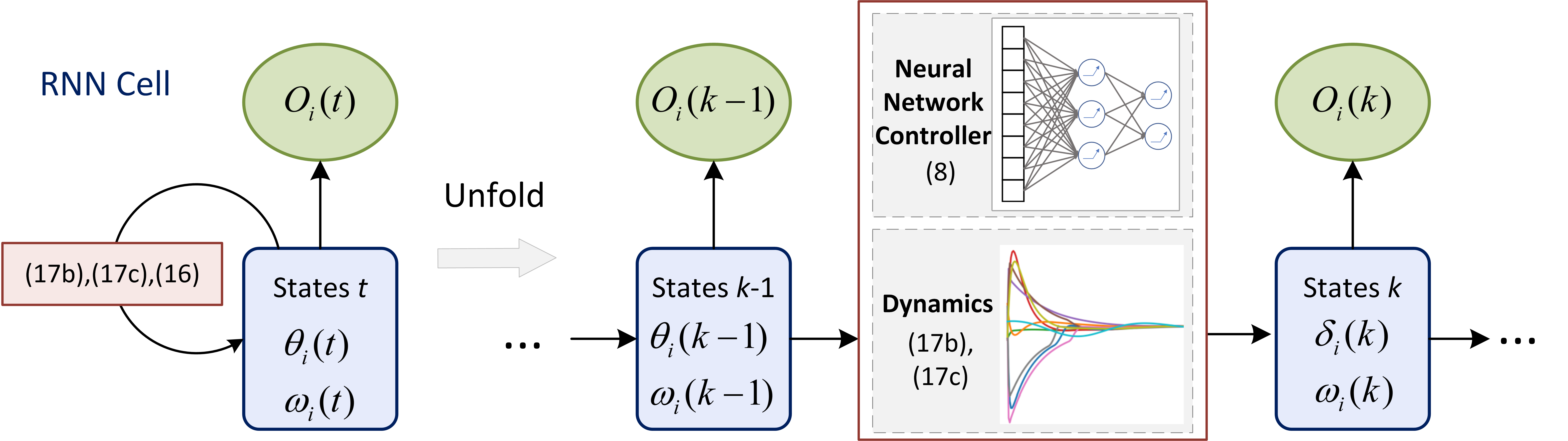}
	\caption{Structure of RNN for the frequency control problem}
	\label{fig:RNN_Structure}
\end{figure}

The operation of RNN is shown by the left side of Fig.~\ref{fig:RNN_Structure}. The cell unit of RNN will remember its current state at the stage $k$ and pass it as an input to the next stage. Unfolding the cell unit through time will give the right side of Fig.~\ref{fig:RNN_Structure}. In this way, RNN can be utilized to deal with time-coupled state variables. Specifically, the state $\left(\theta_i(k-1),\omega_i(k-1)\right)$ for all $i=1,\cdots,n$ at the stage $k-1$ is taken as an input in the state transition function \eqref{subeq:Discrete_optimization_delta} \eqref{subeq:Discrete_optimization_w} and thus the state $\left(\theta_i(k),\omega_i(k)\right)$ for all $i=1,\cdots,n$ at the stage $k$ is obtained. The control function $u_{\varphi_i}(\omega_i(k))$ in the state transition function is formatted through \eqref{eq:u_f+-} to satisfy inequality constraints. The output $O_i(k)=\left [O_i^1(k)\quad O_i^2(k)  \right ]$ at stage $k$ is a vector with two components computed by $O_i^1(k)=\omega_i(k)$ and $O_i^2(k)=\left (u_{\varphi_i}(\omega_i(k))  \right )^2$.
The loss function is formulated to be equivalent with the objective function \eqref{subeq:Discrete_optimization_Obj} as:
\begin{equation}\label{eq:RNN_Loss}
    Loss=\sum_{i=1}^{N} \max_{k=0,\cdots,K} |O_i^1(k)|+\gamma\frac{1}{K}\sum_{k=1}^{K}O_i^2(k)
\end{equation}
 The trainable variables $\bm{\varphi}$ is specified in the neural network controller \eqref{eq:u_f+-} and updated by gradient descent through the Loss function $\eqref{eq:RNN_Loss}$. The unfolded structure of RNN form a directed graph along a temporal sequence where the gradient of Loss function can be efficiently computed by auto-differentiation mechanisms~\cite{griewank1989automatic}.  

\subsection{Algorithm}
The pseudo-code for our proposed method is given in Algorithm 1. The variables to be trained are weights $\bm{\varphi}=\{\hat{\bm{q}},\hat{\bm{b}},\hat{\bm{z}},\hat{\bm{c}}\}$ for control network represented by \eqref{eq:f+}-\eqref{eq:u_f+-} . The $i-th$ row of $\hat{\bm{q}}$ and $\hat{\bm{z}}$ are the vector $\hat{\bm{q}}_i$ and $\hat{\bm{z}}_i$ in \eqref{eq:f+} and \eqref{eq:f-}, respectively. The $i$-th column of $\hat{\bm{b}}$ and $\hat{\bm{c}}$ are the vector $\hat{\bm{b}}_i$ and $\hat{\bm{c}}_i$ in \eqref{eq:f+} and \eqref{eq:f-}, respectively. Training is implemented in a batch updating style where the $h$-th batch initialized with randomly generated initial states $\{\theta_i^h(0),\omega_i^h(0)\}$ for all $i=1,\cdots,n$. The evolution of states in $K$ stages will be computed through the structure of RNN as shown by Fig.~\ref{fig:RNN_Structure}. Adam algorithm is adopted to update weights in each episode.
\\
\begin{algorithm}
 \caption{Reinforcement Learning with RNN}
 \begin{algorithmic}[1]
 \renewcommand{\algorithmicrequire}{\textbf{Require: }}
 \renewcommand{\algorithmicensure}{\textbf{Input:}}
 \REQUIRE  Learning rate $\alpha$, batch size $H$, total time stages K, number of episodes $I$, parameters in optimal frequency control problem \eqref{eq:Discrete_optimization}
 \ENSURE The bound of $\overline{\theta}_i$ and $\overline{\omega}_i$ to generate the initial states\\
\textit{Initialisation} :Initial weights $\bm{\varphi}$ for control network
  \FOR {$episode = 1$ to $I$}
  \STATE  Generate initial states $\theta_i^h(0),\omega_i^h(0)$ for the $i$-th bus in the $h$-th batch, $i=1,\cdots,n$, $h=1,\cdots,H$\\
  \STATE  Reset the state of cells in each batch as the initial value $x_i^h\leftarrow \{\theta_i^h(0),\omega_i^h(0)\}$.\\
  \STATE  RNN cells compute through K stages to obtain the output $\{O_{h,i}(0),O_{h,i}(1),\cdots,O_{h,i}(K)\}$ \\
  \STATE  Calculate total loss of all the batches $Loss=\frac{1}{H}\sum_{h=1}^{H} \sum_{i=1}^{N} \max_{k=0,\cdots,K} |O_{h,i}^1(k)|+\gamma\frac{1}{K}\sum_{k=1}^{K}O_{h,i}^2(k)$\\
  \STATE  Update weights in the neural network by passing $Loss$ to Adam optimizer:
  $\bm{\varphi} \leftarrow \bm{\varphi}-\alpha \text{Adam}(Loss)$
  \ENDFOR
 \end{algorithmic} 
 \end{algorithm}

%% file: simulation.tex
Case studies are conducted on the IEEE New England 10-machine 39-bus (NE39)
power networks to illustrate the effectiveness of the proposed method. Firstly, we show that the proposed Lyapunov-based approaches for designing neural network controller can guarantee stability, while unconstrained neural networks may result in unstable controllers. Then, we show that the proposed structure can learn a nonlinear controller that performs better than other controllers. To ensure that our results apply in practice, simulations are conducted on the system with  $6^{nd}$-order generator model as well as PLL loops on the inverters for frequency measurement~\cite{chow1992toolbox, ortega2015generalized}.

\subsection{Simulation Setting}
We use TensorFlow 2.0 framework to build the reinforcement learning environment and run the training process in Google Colab with a single Nvidia Tesla P100 GPU with 16GB memory. Power System Toolbox (PST) in MATLAB is utilized to simulate the dynamic response from 6-order generator model with turbine-governing system and 2-order phase-locked-loop (PLL) block on the inverter-connected resources
~\cite{chow1992toolbox, ortega2015generalized}. Parameters for the  transient and sub-transient process of generators are obtained in~\cite{demetriou2015dynamic}. 
The system is in the Kron reduced form~\cite{nishikawa2015comparative,Tinu20} and its dynamics is represented by \eqref{eq:Dynamic}. The bound on action $\overline{u}_i$ is generated to be uniformly distributed in $[0.8P_{i},P_{i}]$. 
The initial states of angle and frequency are randomly generated such that $\delta_i(0)$ is uniformly distributed in $[-0.05,0.05]\,\text{rad}$, $\omega_i(0)$ is uniformly distributed in $[-0.1,0.1]\,\text{Hz}$. The cost coefficient $\gamma=0.01$. The stepsize between time states is set as $\Delta t= 0.01s$ and the total time stages is $K=200$.


We compare the performance of the proposed RNN based structure where the neural network controller is designed with and without the Lyapunov-based approach, and the drop control with optimized linear coefficient. The parameter settings are as follows:
\begin{enumerate}
  \item RNN-Lyapunov: Neural network controller designed based on Algorithm 1, which satisfies Theorem
  ~\ref{theorem:Exponential_stable}. The episode number, batch size and the number of neurons are 600, 800 and 20, respectively.
  Parameters of RNN are updated using Adam with learning rate initializes at 0.05 and decays every 30 steps with a base of 0.7.
  \item RNN-Wo-Lyapunov: Controllers are learned without imposing any structures and purely optimizes the reward during training. The controllers are parametrized as neural networks with two dense-layer and the activation function in the first layer is $\tanh$. All the other parameters are the same as RNN-Lyapunov.
  \item Linear droop control: let $k_i$ be the droop coefficient for bus $i$ and the droop control policy is $u_i(\omega_i)=k_i \omega_i$ for $i=1,\dots,n$, thresholded to their upper and lower bounds.  The optimized droop coefficient is obtained by solving~\eqref{eq:Discrete_optimization} using fmincon function
of Matlab.
  \item PG-Monotone: This controller is to demonstrate the performance improvements of using RNN during training. So here we impose the stacked-ReLU structure and trained with REINFORCE Policy Gradient algorithm~\cite{sutton2018reinforcement}. The neural networks for controller, the episode number, batch number and optimizer are the same as RNN-Lyapunov. The learning rate initializes at 0.01 and decays every 30 steps with a base of 0.7.
\end{enumerate}

\subsection{Necessity of Lyapunov-based Approach} 
\begin{figure}[ht]	
	\centering
	\includegraphics[scale=0.45]{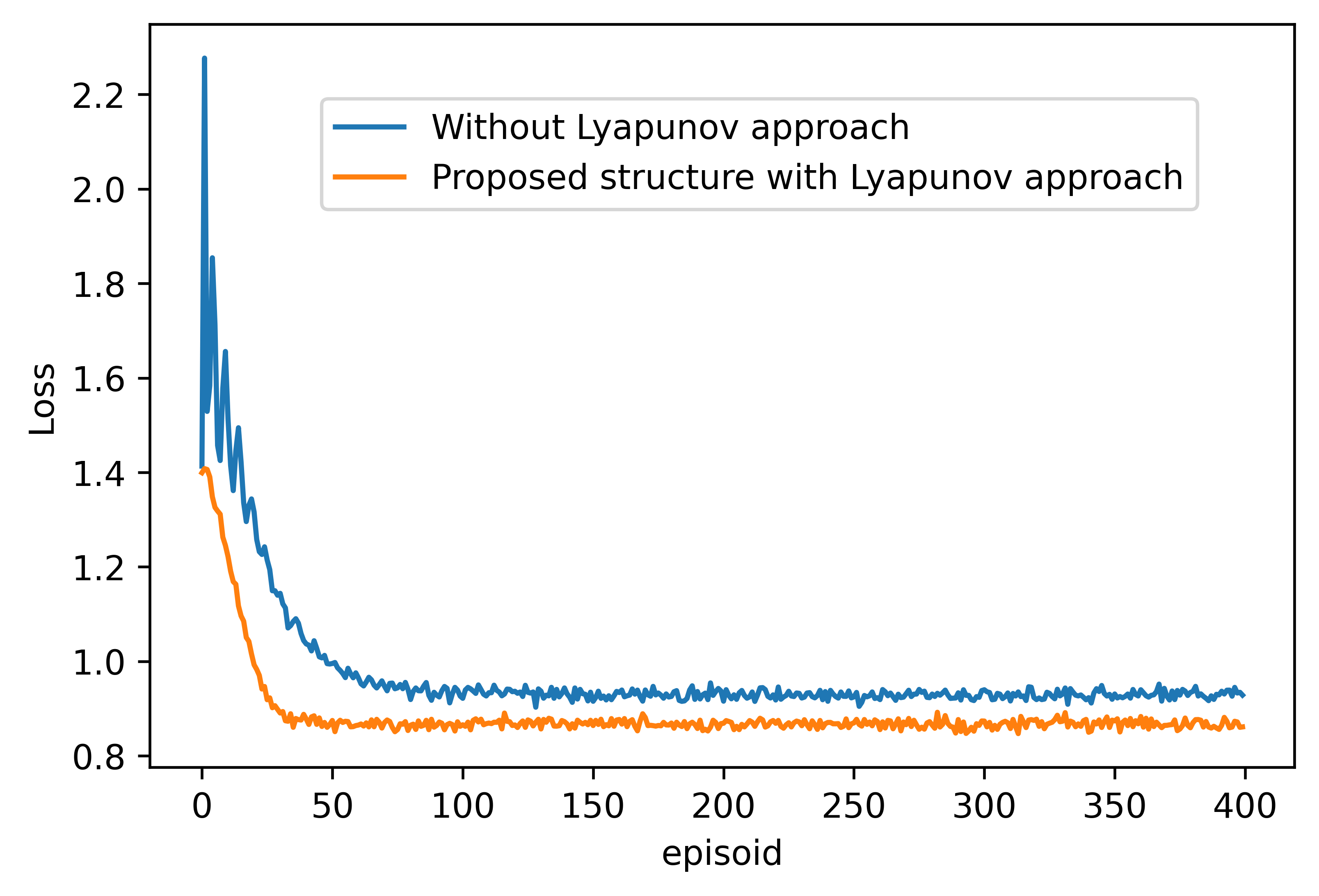}
	\caption{\highlights{Average batch loss along episodes for neural network controller designed with and without the Lyapunov-based approach. Both converges, with the former converging much for quickly than the latter.} }
	\label{fig:Loss_epi}
\end{figure}

\begin{figure}[ht]
\centering
\subfloat[Dynamics of $\omega$ (left) and $\delta$ (right) for RNN-Lyapunov]{\includegraphics[width=3.3in]{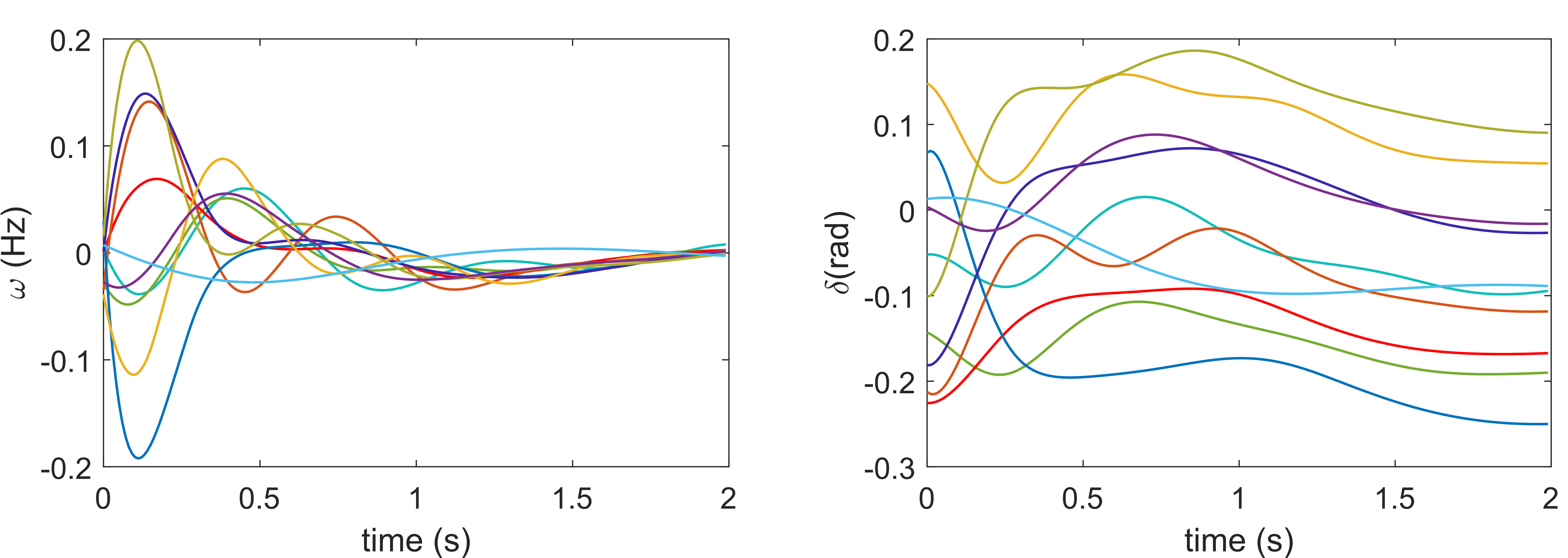}%
\label{fig_first_case}}
\hfil
\subfloat[Dynamics of $\omega$ (left) and $\delta$ (right) for RNN-Wo-Lyapunov]{\includegraphics[width=3.3in]{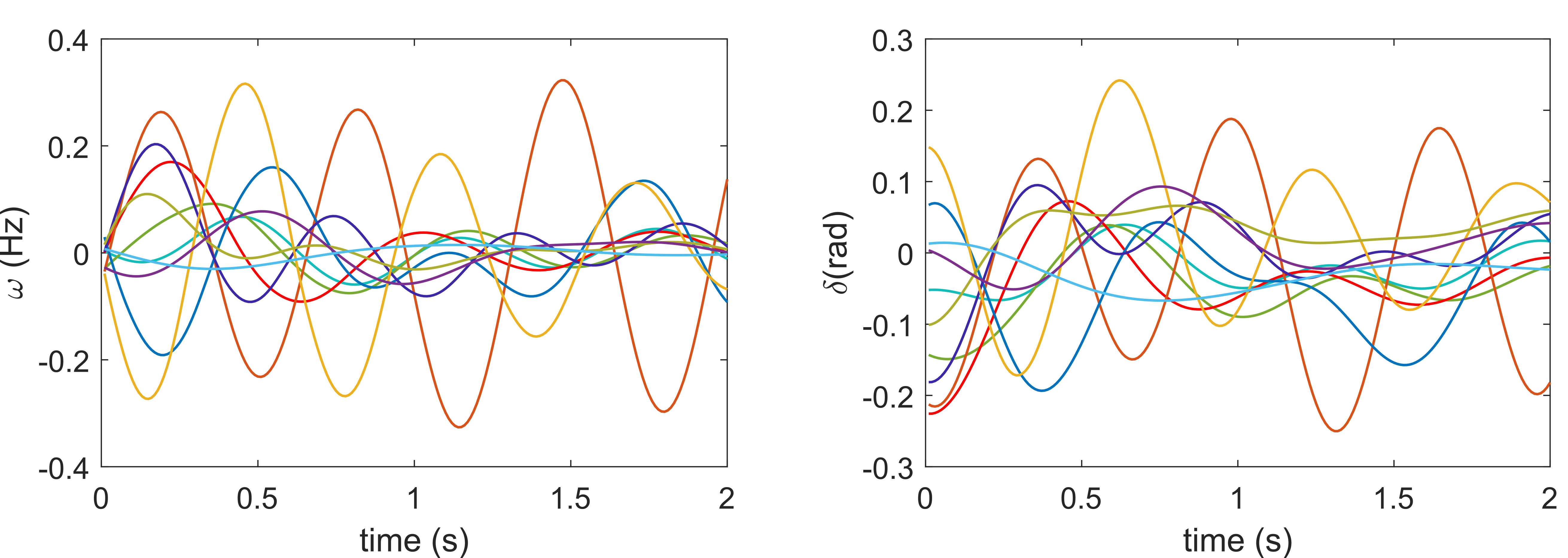}%
\label{fig_second_case}}
\caption{Dynamics of angle $\delta$ and frequency deviation $\omega$ in 10 generator buses corresponding to (a) the neural network controller designed with the Lyapunov-based approach and (b) the neural network controller designed without the Lyapunov-based approach. The two controllers exhibit qualitatively different behavior even though they both achieve finite training losses in Fig.~\ref{fig:Loss_epi}. The controller designed without the Lyapunov-based approach leads to unstable trajectories of the system.}
\label{fig:Dynamic_Behavior_bad}
\end{figure}

Theorem
~\ref{theorem:Exponential_stable} ensures that the learned controller would be locally exponentially stable, but it's interesting to check the performance of an unconstrained controller. Intuitively, an unstable controller should lead to large costs since some trajectories would be blowing up. Then maybe a controller that minimizes the cost would also be stabilizing. 

Figure~\ref{fig:Loss_epi} shows the training loss between controllers
learned with and without the Lyapunov-based approach. Both
losses converge, with the Lyapunov-based controller having
better performances. However, when we implement the controllers,
the one without considering stability is unstable
and leads to very large state oscillations (Fig.~\ref{fig_second_case}). In contrast,
the controller constrained by the Lyapunov condition shows
good performance (Fig.~\ref{fig_first_case}). The reason for this dichotomy
in performance is that we can only check a finite number
of trajectories during training, and good training performance
does not in itself guarantee good generalization. Therefore,
explicitly constraining the controller structure is necessary.


\subsection{Performance Comparisons}
This subsection shows that the proposed method can learn a static nonlinear controller that outperforms the optimal linear droop controller and the RNN training technique is much more efficient than using a standard policy gradient method. Figure~\ref{fig:Action_r5_5} illustrates the control policy learnt from RNN-Lyapunov, Policy Gradient and the linear droop control with optimized droop coefficient for four generators. Compared with the traditional droop control, the proposed stacked-ReLU neural network learns a nonlinear controller with different shapes for RNN-Lyapunov and PG-Monotone.

We first study the learned controllers and their performances during a sudden change in load/generation. Suppose bus $4$ experiences a step load increase of 0.05 p.u. occurs at t=0.3s and a step load recovery occurs at t=5.3s. Figure~\ref{fig:Compare_RNN_Linear_r5_5} illustrates the dynamics of $\omega$ and corresponding control
action $u$ under each of the controllers. 
After the step load change, RNN-Lyapunov and linear droop control  achieve similar maximum frequency deviation, while the control action of RNN-Lyapunov is much lower than the other. PG-Monotone shows higher frequency deviation and oscillations. Therefore, the proposed RNN-Lyapunov approach has the minimal cost.
The computational time of the proposed RNN based method is 1080.38s, while the computational time of REINFORCE policy gradient takes 4206.43s. Therefore, the proposed RNN based structure reduces computational time by approximate 74.32\% compared with the general RL structure.

\begin{figure}[t]	
  	\centering
  	\includegraphics[width=3.4in]{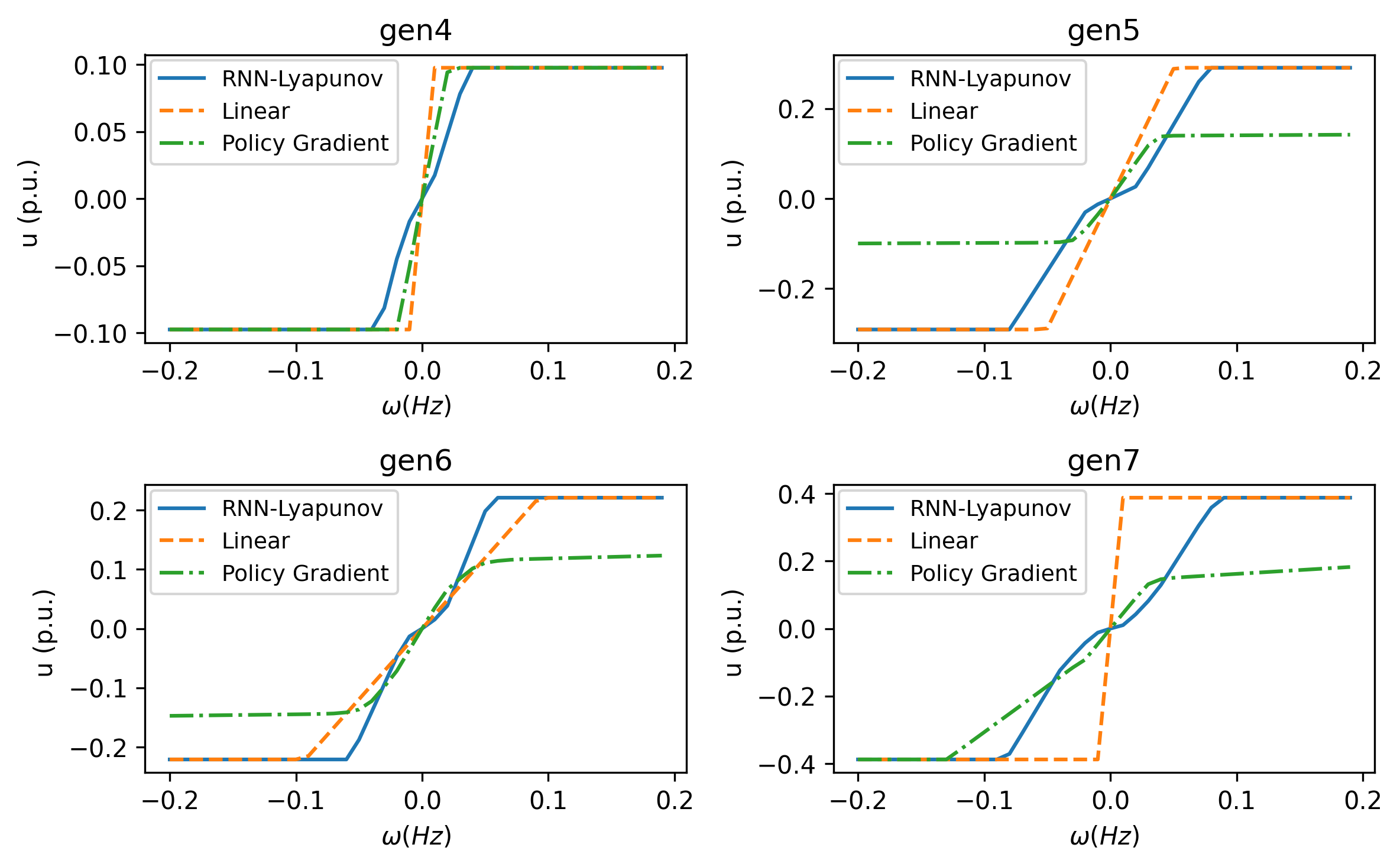}
  	\caption{Examples of learned controller $u$ corresponding to RNN-Lyapunov, Linear droop control and Policy Gradient for generator buses 4,5,6 and 7. The comparison shows that the proposed Stacked-ReLU neural network learns  nonlinear controllers in flexible shapes. }
  	\label{fig:Action_r5_5}
  \end{figure}

\begin{figure}[ht]
\centering
\subfloat[Dynamics of $\omega$ (left) and $u$ (right) for RNN-Lyapunov ]{\includegraphics[width=3.4in]{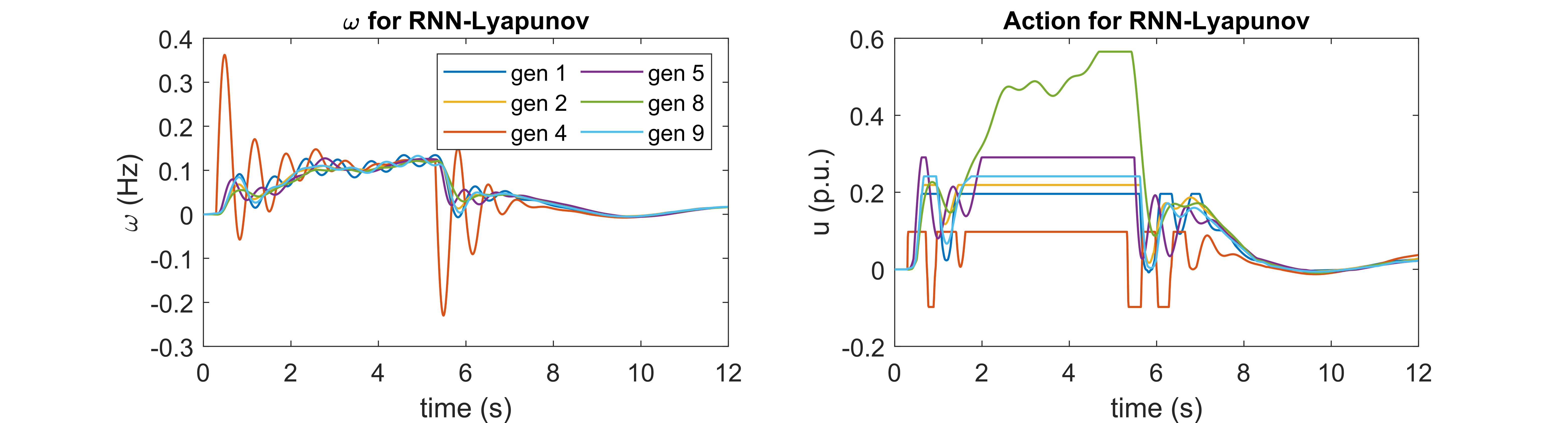}%
}
\hfil
\subfloat[Dynamics of $\omega$ (left) and $u$ (right) for linear droop control]{\includegraphics[width=3.4in]{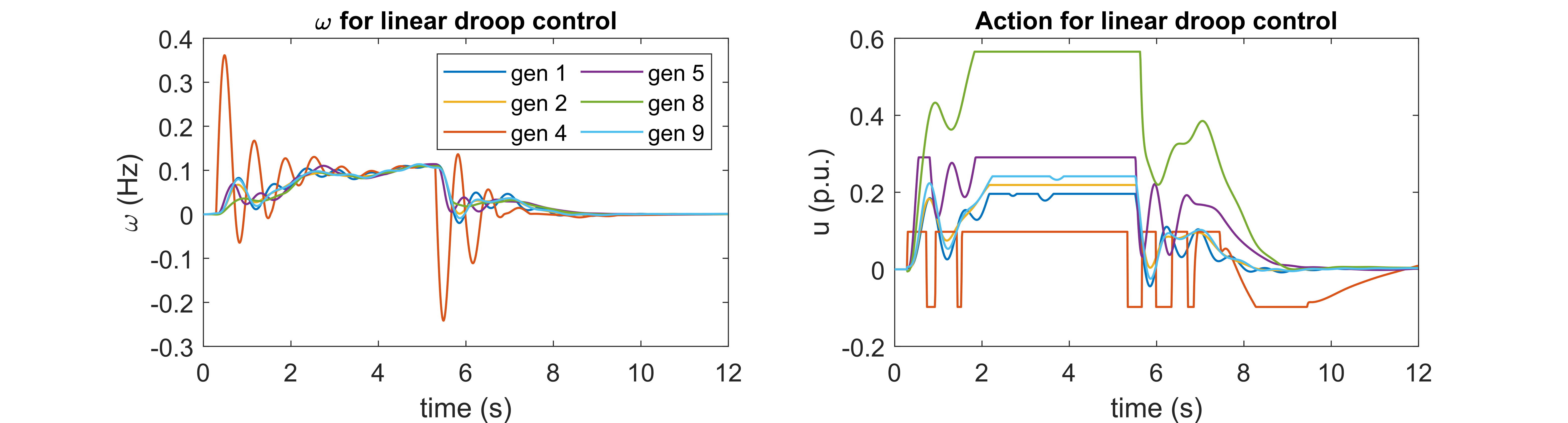}%
}
\hfil
\subfloat[Dynamics of $\omega$ (left) and $u$ (right) for controller obtained by PG-Monotone]{\includegraphics[width=3.4in]{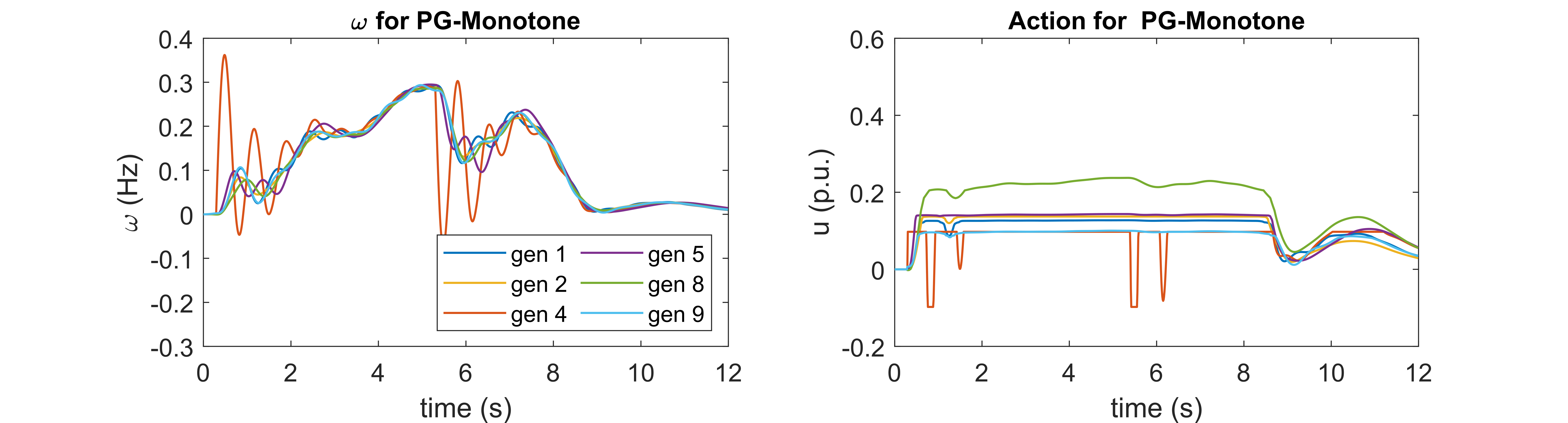}%
}
\caption{Dynamics of  the frequency deviation $w$ and the control action $u$ in selected generator buses corresponding to (a) Lyapunov-guided neural network controller learned with RNN. (b) Linear droop control. (c) Lyapunov-guided neural network controller learned from Policy Gradient with Monotone structure design. The proposed RNN controller has the smallest cost.}
\label{fig:Compare_RNN_Linear_r5_5}
\end{figure}

Next, we randomize the initial starting points to simulate and test the performance of the three methods under multiple different trajectories. We fix initial $\delta$ in $U[-0.1,0.1]\,\text{rad}$ and let the initial $\omega$ to uniformly distributed in $U[-\bar{\omega},\bar{\omega}]$ around the equilibrium, where $\bar{\omega}$ denotes the variation bound of initial $\omega$. 
The average loss corresponding to $\bar{\omega} =0.0,0.025,\cdots,0.15\,\text{Hz}$ are illustrated in Fig.~\ref{fig:Loss_r5_5}. $\bar{\omega}=0$ is the case that no variation of $\omega$ exists in the initial condition. 
Overall, RNN-Lyapunov remains approximate 11.39\% and 5.41\% lower in average loss than that of linear droop control, PG-Monotone, respectively. Therefore, the proposed method learn the nonlinear controller that leads to better average control performance under different initial conditions.

\begin{figure}[ht]	
	\centering
	\includegraphics[scale=0.45]{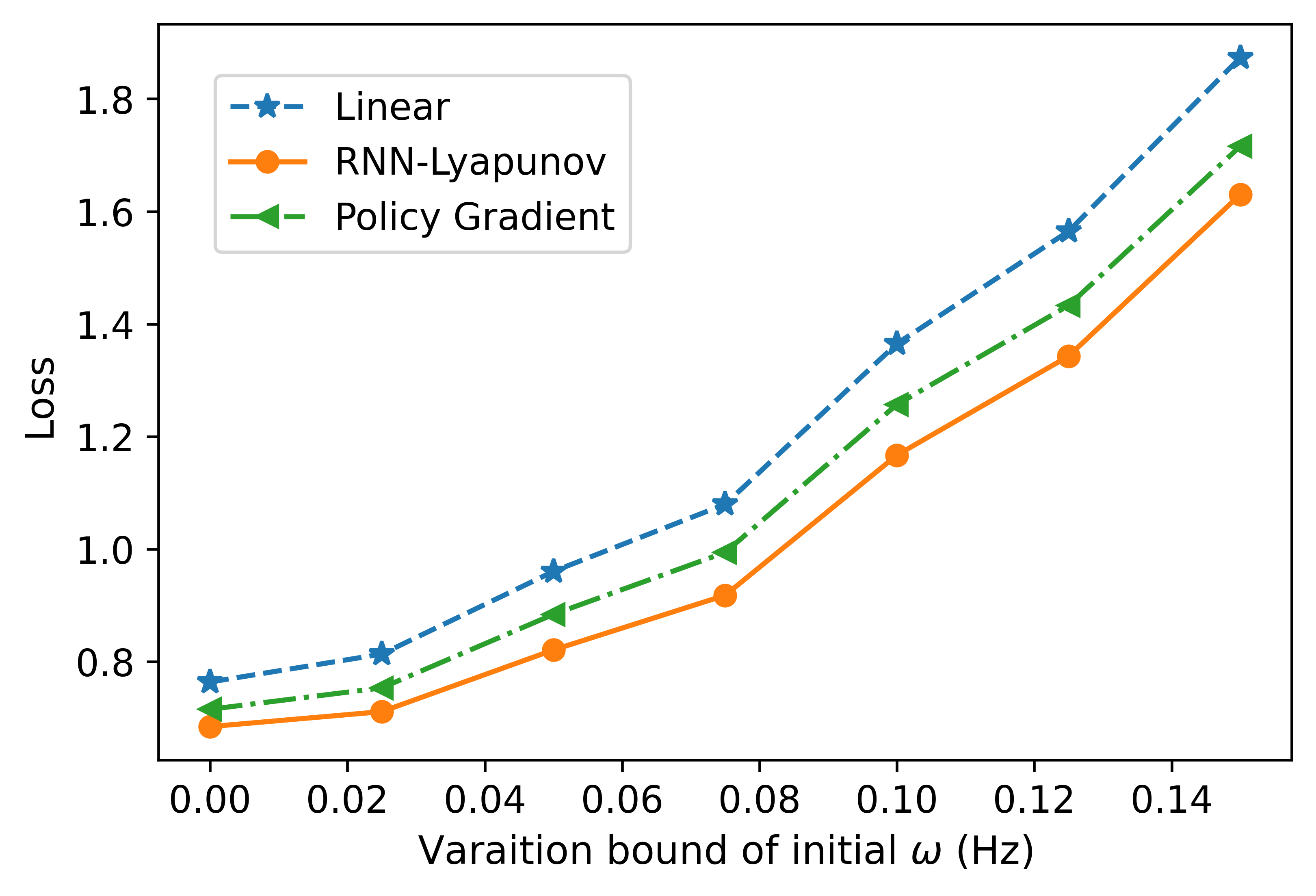}
	\caption{Loss with different variation range of initial conditions for RNN-Lyapunov, Linear droop controller and Policy Gradient. Compared with Linear droop controller and PG-Monotone, RNN-Lyapunov reduces the loss by approximate 11.39\%, 5.41\%, respectively. }
	\label{fig:Loss_r5_5}
\end{figure}

%% file: appendix.tex

\section{Proof of Lemma~\ref{lemma:Pos_Lyapunov}}\label{appendices:lemma_Pos_Lyapunov}
\begin{proof}
The proof is similar to the one of~\cite[Lemma~14]{weitenberg2018robust}, which bounds $V(\boldsymbol{\delta}, \boldsymbol{\omega})$ term by term. Firstly, using the Rayleigh-Ritz theorem~\cite{Horn2012MA}, the kinetic energy term, $\frac{1}{2}\sum_{i=1}^{n}M_i(\omega_{i}-\omega^\ast)^{2}$, is lower bounded by $\frac{1}{2}\lambda_\mathrm{min}(\boldsymbol{M})||\boldsymbol{\omega}\!-\!\boldsymbol{\omega}^\ast||^2_2$ and upper bounded by $\frac{1}{2}\lambda_\mathrm{max}(\boldsymbol{M})||\boldsymbol{\omega}\!-\!\boldsymbol{\omega}^\ast||^2_2$.
Then, with a direct application of~\cite[Lemma~4]{weitenberg2018exponential}, the potential energy term $W_\mathrm{p}(\boldsymbol{\delta})$ in \eqref{eq:Wp} can be bounded by $\beta_{1}\left\|\boldsymbol{\delta}-\boldsymbol{\delta}^{\ast}\right\|_2^2 \leq W_\mathrm{p}(\boldsymbol{\delta}) \leq \beta_{2}\left\|\boldsymbol{\delta}-\boldsymbol{\delta}^{\ast}\right\|_2^2 $
for some constants $\beta_{1}>0$ and $\beta_{2}>0$. 

To deal with the cross term $W_\mathrm{c}(\boldsymbol{\delta})$, we define $p_{\mathrm{e},i}(\boldsymbol{\delta}):=\sum_{j=1}^{n} B_{ij} \sin(\delta_{ij})$. Then, $W_\mathrm{c}(\boldsymbol{\delta}) =\left(\boldsymbol{p}_\mathrm{e}(\boldsymbol{\delta})\!-\!\boldsymbol{p}_\mathrm{e}(\boldsymbol{\delta}^\ast)\right)^T\!\boldsymbol{M}(\boldsymbol{\omega}\!-\!\boldsymbol{\omega}^\ast)$. Clearly, $-|W_\mathrm{c}(\boldsymbol{\delta})|\leq W_\mathrm{c}(\boldsymbol{\delta})\leq|W_\mathrm{c}(\boldsymbol{\delta})|$. For $\forall \boldsymbol{x},\boldsymbol{y}\in\real^{n}$, $2|\boldsymbol{x}^T\boldsymbol{y}|\leq||\boldsymbol{x}||^2_2+||\boldsymbol{y}||^2_2$. Thus, we have
\begin{align*}
|W_\mathrm{c}(\boldsymbol{\delta})|\leq&\ \dfrac{1}{2}\left(||\boldsymbol{p}_\mathrm{e}(\boldsymbol{\delta})\!-\!\boldsymbol{p}_\mathrm{e}(\boldsymbol{\delta}^\ast)||^2_2+||\boldsymbol{M}(\boldsymbol{\omega}\!-\!\boldsymbol{\omega}^\ast)||^2_2\right)\nonumber \\
\leq&\ \dfrac{1}{2}\left(\gamma_2||\boldsymbol{\delta}\!-\!\boldsymbol{\delta}^\ast||^2_2+\lambda_\mathrm{max}(\boldsymbol{M})^2||\boldsymbol{\omega}\!-\!\boldsymbol{\omega}^\ast||^2_2\right)\,,
\end{align*}
where  the second inequality comes from~\cite[Lemma~4]{weitenberg2018exponential} and the Rayleigh-Ritz theorem, with some $\gamma_2>0$. 
Hence, the $W_\mathrm{c}(\boldsymbol{\delta})$ term is lower bounded by $-\dfrac{1}{2}\left(\gamma_2||\boldsymbol{\delta}\!-\!\boldsymbol{\delta}^\ast||^2_2+\lambda_\mathrm{max}(\boldsymbol{M})^2||\boldsymbol{\omega}\!-\!\boldsymbol{\omega}^\ast||^2_2\right)$ and upper bounded by $\dfrac{1}{2}\left(\gamma_2||\boldsymbol{\delta}\!-\!\boldsymbol{\delta}^\ast||^2_2+\lambda_\mathrm{max}(\boldsymbol{M})^2||\boldsymbol{\omega}\!-\!\boldsymbol{\omega}^\ast||^2_2\right)$.
Finally, combining the inequalities, we can bound the entire Lyapunov function $V(\boldsymbol{\delta}, \boldsymbol{\omega})$ in \eqref{eq: Lyapunov} with
\begin{align*}
\alpha_{1} &:=\dfrac{1}{2}\min \left(\lambda_\mathrm{min}(\boldsymbol{M})-\epsilon \lambda_\mathrm{max}(\boldsymbol{M})^2, 2\beta_{1}-\epsilon \gamma_{2}\right)>0\,, \\
\alpha_{2} &:=\dfrac{1}{2}\max \left(\lambda_\mathrm{max}(\boldsymbol{M})+\epsilon \lambda_\mathrm{max}(\boldsymbol{M})^2, 2\beta_2+\epsilon \gamma_{2}\right)>0\,,
\end{align*}
for sufficiently small $\epsilon>0$.
\end{proof}

\section{Proof of Lemma~\ref{lemma: dot_V}}\label{appendices:lemma_dot_V}
\begin{proof}
We start by computing the partial derivatives of $V(\boldsymbol{\delta}, \boldsymbol{\omega})$ with respect to each state, i.e.,
\begin{subequations}\label{eq: Lyapunov_derivative_partial
}
\begin{align}
\frac{\partial V }{\partial \delta_i}&=p_{\mathrm{e},i}(\boldsymbol{\delta})-p_{\mathrm{e},i}(\boldsymbol{\delta}^\ast) +\epsilon\!\!\sum_{j=1,j\neq i}^nB_{ij}\cos(\delta_{ij})M_i(\omega_i-\omega^\ast)\nonumber\\
&\quad-\epsilon\!\!\sum_{j=1,j\neq i}^n B_{ij}\cos(\delta_{ij})M_j(\omega_j-\omega^\ast)\,,\nonumber\\
\frac{\partial V }{\partial \omega_i}&=M_i\left[\omega_i-\omega^\ast+\epsilon\left(p_{\mathrm{e},i}(\boldsymbol{\delta})-p_{\mathrm{e},i}(\boldsymbol{\delta}^\ast)\right)\right]\,.\nonumber
\end{align}
\end{subequations}


Therefore, the time derivative of $V(\boldsymbol{\delta}, \boldsymbol{\omega})$, i.e.,  $\dot{V}(\boldsymbol{\delta}, \boldsymbol{\omega})$, is 
\begin{equation*} 
 \begin{aligned} 
& \sum_{i=1}^{n}\left(\frac{\partial V }{\partial \delta_i} \dot{\delta}_i+\frac{\partial V }{\partial \omega_i}\dot{\omega}_i\right)
\\
=&\ \!\!\left(\boldsymbol{p}_\mathrm{e}(\boldsymbol{\delta})\!-\!\boldsymbol{p}_\mathrm{e}(\boldsymbol{\delta}^\ast)+\epsilon\boldsymbol{H}(\boldsymbol{\delta})\boldsymbol{M}(\boldsymbol{\omega}\!-\!\boldsymbol{\omega}^\ast)\right)^T\left(\boldsymbol{\omega}-\bm{1}\frac{\bm{1}^T\boldsymbol{\omega}}{n}\right)\\
&\ \!\!+\!\left[\boldsymbol{\omega}\!-\!\boldsymbol{\omega}^\ast \!\!+\! \epsilon\left(\boldsymbol{p}_\mathrm{e}(\boldsymbol{\delta})\!-\!\boldsymbol{p}_\mathrm{e}(\boldsymbol{\delta}^\ast)\right)\right]^T\!\!\left(\boldsymbol{p}_\mathrm{m}\!-\!\boldsymbol{D}\boldsymbol{\omega}\!-\!\boldsymbol{u}(\boldsymbol{\omega})\!-\!\boldsymbol{p}_\mathrm{e}(\boldsymbol{\delta})\right)\\
&+\underbrace{\left(\boldsymbol{p}_\mathrm{e}(\boldsymbol{\delta})\!-\!\boldsymbol{p}_\mathrm{e}(\boldsymbol{\delta}^\ast)\!+\!\epsilon\boldsymbol{H}(\boldsymbol{\delta})\boldsymbol{M}(\boldsymbol{\omega}\!-\!\boldsymbol{\omega}^\ast)\right)^T\left(\bm{1}\frac{\bm{1}^T\boldsymbol{\omega}}{n}\!-\!\bm{1}\omega^\ast\right)}_{=0}\\
&\ \!\!-\!\left[\boldsymbol{\omega}\!-\!\boldsymbol{\omega}^\ast \!\!+\! \epsilon\!\left(\boldsymbol{p}_\mathrm{e}(\boldsymbol{\delta})\!-\!\boldsymbol{p}_\mathrm{e}(\boldsymbol{\delta}^\ast)\right)\right]^T\!\!\underbrace{\left(\!\boldsymbol{p}_\mathrm{m}\!\!-\!\!\boldsymbol{D}\boldsymbol{\omega}^\ast\!\!\!-\!\boldsymbol{u}(\boldsymbol{\omega}^\ast)\!-\!\boldsymbol{p}_\mathrm{e}(\boldsymbol{\delta}^\ast)\!\right)}_{=\boldsymbol{0}}
\\
=&\ \epsilon\left(\boldsymbol{H}(\boldsymbol{\delta})\boldsymbol{M}(\boldsymbol{\omega}\!-\!\boldsymbol{\omega}^\ast)\right)^T(\boldsymbol{\omega}\!-\!\boldsymbol{\omega}^\ast)-(\boldsymbol{\omega}\!-\!\boldsymbol{\omega}^\ast)^T\boldsymbol{D}(\boldsymbol{\omega}\!-\!\boldsymbol{\omega}^\ast)\\
&-\epsilon\left(\boldsymbol{p}_\mathrm{e}(\boldsymbol{\delta})\!-\!\boldsymbol{p}_\mathrm{e}(\boldsymbol{\delta}^\ast)\right)^T\left(\boldsymbol{p}_\mathrm{e}(\boldsymbol{\delta})\!-\!\boldsymbol{p}_\mathrm{e}(\boldsymbol{\delta}^\ast)\right)\\&-\epsilon\left(\boldsymbol{p}_\mathrm{e}(\boldsymbol{\delta})\!-\!\boldsymbol{p}_\mathrm{e}(\boldsymbol{\delta}^\ast)\right)^T\boldsymbol{D}(\boldsymbol{\omega}\!-\!\boldsymbol{\omega}^\ast)\\
&-\left[\boldsymbol{\omega}\!-\!\boldsymbol{\omega}^\ast \!\!+\! \epsilon\left(\boldsymbol{p}_\mathrm{e}(\boldsymbol{\delta})\!-\!\boldsymbol{p}_\mathrm{e}(\boldsymbol{\delta}^\ast)\right)\right]^T\left(\boldsymbol{u}(\boldsymbol{\omega})-\boldsymbol{u}(\boldsymbol{\omega}^\ast)\right)\,,
\end{aligned}
\end{equation*}
which is exactly \eqref{eq:Vdot}. Note that the extra terms in the second equality are added to construct a quadratic format without affecting the the original value of  $\dot{V}(\boldsymbol{\delta}, \boldsymbol{\omega})$ since $\boldsymbol{p}_\mathrm{e}(\boldsymbol{\delta})^T\bm{1}=\bm{0}$, $\boldsymbol{H}(\boldsymbol{\delta})^T\bm{1}=\bm{0}$, and $\boldsymbol{p}_\mathrm{m}\!-\!\boldsymbol{D}\boldsymbol{\omega}^\ast\!-\!\boldsymbol{u}(\boldsymbol{\omega}^\ast)\!-\!\boldsymbol{p}_\mathrm{e}(\boldsymbol{\delta}^\ast)=0$ by the condition at the equilibrium given in \eqref{eq:equli-sync-delta}. 

It remains to show that $\boldsymbol{Q}(\boldsymbol{\delta})\succ0$, which follows directly from the fact that the Schur complement of the block $\epsilon \boldsymbol{I}$ in $\boldsymbol{Q}(\boldsymbol{\delta})$ is positive definite: $\boldsymbol{D}-\dfrac{\epsilon}{2}(\boldsymbol{H}(\boldsymbol{\delta})\boldsymbol{M}+\boldsymbol{M}\boldsymbol{H}(\boldsymbol{\delta}))-\frac{\epsilon}{4} \boldsymbol{D}^2\succ 0$ for sufficiently small $\epsilon$.
\end{proof}

\section{Proof of Theorem~\ref{theorem:Universial Approximation}}\label{appendices:stacked-ReLU}
Let $\alpha$ bound the magnitude of first derivative of $r$ on $\mathbb{X}$ . Define an equispaced grid of points on $\mathbb{X}$, where $\beta=\frac{1}{n}$ is the spacing between grid points along each dimension. Corresponding to each grid interval $[k\beta,(k+1)\beta]$, assign a linear function $y(x)=r(k\beta)+\frac{r((k+1)\beta)-r(k\beta)}{\beta}(x-k\beta)$, where $y(k\beta)=r(k\beta)$ and $y((k+1)\beta)=r((k+1)\beta)$.  For all $x\in [k\beta,(k+1)\beta]$, from monotonic property, we have $r(k\beta)\leq r(x) \leq r((k+1)\beta)$ and $r(k\beta)\leq y(x) \leq r((k+1)\beta)$. Therefore, we can bound the approximation error by
\begin{equation}
    |y(x)-r(x)|\leq |r((k+1)\beta)-r(k\beta)|
\end{equation}
By mean value theorem, we know that 
\begin{equation}
r((k+1)\beta)-r(k\beta)=\beta  \frac{\partial r(c)}{\partial x}    
\end{equation}
for some point $c$ on the line segment between $k\beta$ and $(k+1)\beta$. Given the assumptions made at the outset, $|\frac{\partial r(c)}{\partial x}|$ is bounded by $\alpha$ and therefore  $|y(x)-r(x)|$ can be bounded
by $\beta \alpha .$ 

Further, we show that any piece-wise linear function of $y(x)=r(k\beta)+\frac{r((k+1)\beta)-r(k\beta)}{\beta}(x-k\beta)$ can be represented by the proposed construction \eqref{eq:f+}\eqref{eq:f-}. Without loss of generosity, assume that $y(x)$ is the positive part and approximated by $f^+(x)$. Let $b_i^{1}=0$, $q^{1}=r(\beta)$ and subsequently $b_i^{k}=(k-1)\beta$, $\sum_{j=1}^{k} q^{j}=\frac{r(k\beta)-r((k-1)\beta)}{\beta}$ for $k=2,3,\cdots,n.$ Then the construction of $f^+(x)$ through $\eqref{eq:f+}$ is exactly the same as $y(x)$. Therefore,  $|f(x)-r(x)|$ can also be bounded
by $\beta \alpha .$ We take $\beta<\frac{\epsilon}{\alpha}$ to complete the proof.

%% file: main.bbl
\begin{thebibliography}{10}
\providecommand{\url}[1]{#1}
\csname url@samestyle\endcsname
\providecommand{\newblock}{\relax}
\providecommand{\bibinfo}[2]{#2}
\providecommand{\BIBentrySTDinterwordspacing}{\spaceskip=0pt\relax}
\providecommand{\BIBentryALTinterwordstretchfactor}{4}
\providecommand{\BIBentryALTinterwordspacing}{\spaceskip=\fontdimen2\font plus
\BIBentryALTinterwordstretchfactor\fontdimen3\font minus
  \fontdimen4\font\relax}
\providecommand{\BIBforeignlanguage}[2]{{%
\expandafter\ifx\csname l@#1\endcsname\relax
\typeout{** WARNING: IEEEtran.bst: No hyphenation pattern has been}%
\typeout{** loaded for the language `#1'. Using the pattern for}%
\typeout{** the default language instead.}%
\else
\language=\csname l@#1\endcsname
\fi
#2}}
\providecommand{\BIBdecl}{\relax}
\BIBdecl

\bibitem{kroposki2017achieving}
B.~Kroposki, B.~Johnson, Y.~Zhang, V.~Gevorgian, P.~Denholm, B.-M. Hodge, and
  B.~Hannegan, ``Achieving a 100\% renewable grid: Operating electric power
  systems with extremely high levels of variable renewable energy,'' \emph{IEEE
  Power and Energy Magazine}, vol.~15, no.~2, pp. 61--73, 2017.

\bibitem{kundur1994power}
P.~Kundur, N.~J. Balu, and M.~G. Lauby, \emph{Power system stability and
  control}.\hskip 1em plus 0.5em minus 0.4em\relax McGraw-hill New York, 1994,
  vol.~7.

\bibitem{poolla2017optimal}
B.~K. Poolla, S.~Bolognani, and F.~Dorfler, ``Optimal placement of virtual
  inertia in power grids,'' \emph{IEEE Transactions on Automatic Control},
  2017.

\bibitem{zhang2020modeling}
Z.~Zhang, E.~Du, F.~Teng, N.~Zhang, and C.~Kang, ``Modeling frequency dynamics
  in unit commitment with a high share of renewable energy,'' \emph{IEEE
  Transactions on Power Systems}, 2020.

\bibitem{zhao2014design}
C.~Zhao, U.~Topcu, N.~Li, and S.~Low, ``Design and stability of load-side
  primary frequency control in power systems,'' \emph{IEEE Transactions on
  Automatic Control}, vol.~59, no.~5, pp. 1177--1189, 2014.

\bibitem{mallada2017optimal}
E.~Mallada, C.~Zhao, and S.~Low, ``Optimal load-side control for frequency
  regulation in smart grids,'' \emph{IEEE Transactions on Automatic Control},
  vol.~62, no.~12, pp. 6294--6309, 2017.

\bibitem{Tinu20}
A.~{Ademola-Idowu} and B.~{Zhang}, ``Frequency stability using inverter power
  control in low-inertia power systems,'' \emph{IEEE Transactions on Power
  Systems}, pp. 1--1, 2020.

\bibitem{johnson2013synchronization}
B.~B. Johnson, S.~V. Dhople, A.~O. Hamadeh, and P.~T. Krein, ``Synchronization
  of parallel single-phase inverters with virtual oscillator control,''
  \emph{IEEE Transactions on Power Electronics}, vol.~29, no.~11, pp.
  6124--6138, 2013.

\bibitem{MPCFastFrequency}
O.~{Stanojev}, U.~{Markovic}, P.~{Aristidou}, G.~{Hug}, D.~S. {Callaway}, and
  E.~{Vrettos}, ``Mpc-based fast frequency control of voltage source converters
  in low-inertia power systems,'' \emph{IEEE Transactions on Power Systems},
  pp. 1--1, 2020.

\bibitem{chen2021reinforcement}
X.~Chen, G.~Qu, Y.~Tang, S.~Low, and N.~Li, ``Reinforcement learning for
  decision-making and control in power systems: Tutorial, review, and vision,''
  \emph{arXiv preprint arXiv:2102.01168}, 2021.

\bibitem{yan2018data}
Z.~Yan and Y.~Xu, ``Data-driven load frequency control for stochastic power
  systems: A deep reinforcement learning method with continuous action
  search,'' \emph{IEEE Transactions on Power Systems}, vol.~34, no.~2, pp.
  1653--1656, 2018.

\bibitem{qu2020scalable}
G.~Qu, A.~Wierman, and N.~Li, ``Scalable reinforcement learning of localized
  policies for multi-agent networked systems,'' in \emph{Learning for Dynamics
  and Control}.\hskip 1em plus 0.5em minus 0.4em\relax PMLR, 2020, pp.
  256--266.

\bibitem{huang2019adaptive}
Q.~Huang, R.~Huang, W.~Hao, J.~Tan, R.~Fan, and Z.~Huang, ``Adaptive power
  system emergency control using deep reinforcement learning,'' \emph{IEEE
  Transactions on Smart Grid}, 2019.

\bibitem{sutton2018reinforcement}
R.~S. Sutton and A.~G. Barto, \emph{Reinforcement learning: An
  introduction}.\hskip 1em plus 0.5em minus 0.4em\relax MIT press, 2018.

\bibitem{ZHANG2019274}
Y.~Zhang and J.~Cort{\'e}s, ``Distributed transient frequency control for power
  networks with stability and performance guarantees,'' \emph{Automatica}, vol.
  105, pp. 274--285, 2019.

\bibitem{dominguez2012models}
A.~D. Dom{\'\i}nguez-Garc{\'\i}a, ``Models for impact assessment of wind-based
  power generation on frequency control,'' in \emph{Control and Optimization
  Methods for Electric Smart Grids}.\hskip 1em plus 0.5em minus 0.4em\relax
  Springer, 2012, pp. 149--165.

\bibitem{Xu18}
B.~Xu, Y.~Shi, D.~S. Kirschen, and B.~Zhang, ``Optimal battery participation in
  frequency regulation markets,'' \emph{IEEE Transactions on Power Systems},
  vol.~33, no.~6, pp. 6715--6725, 2018.

\bibitem{HG19}
P.~Hidalgo-Gonzalez, R.~Henriquez-Auba, D.~S. Callaway, and C.~J. Tomlin,
  ``Frequency regulation using data-driven controllers in power grids with
  variable inertia due to renewable energy,'' in \emph{2019 IEEE Power Energy
  Society General Meeting (PESGM)}, 2019, pp. 1--5.

\bibitem{jiang2021storage}
Y.~Jiang, E.~Cohn, P.~Vorobev, and E.~Mallada, ``Storage-based frequency
  shaping control,'' \emph{IEEE Transactions on Power Systems}, 2021.

\bibitem{sauer2017power}
P.~W. Sauer, M.~A. Pai, and J.~H. Chow, \emph{Power system dynamics and
  stability: with synchrophasor measurement and power system toolbox}.\hskip
  1em plus 0.5em minus 0.4em\relax John Wiley \& Sons, 2017.

\bibitem{chow1992toolbox}
J.~H. Chow and K.~W. Cheung, ``A toolbox for power system dynamics and control
  engineering education and research,'' \emph{IEEE transactions on Power
  Systems}, vol.~7, no.~4, pp. 1559--1564, 1992.

\bibitem{tabas2019optimal}
D.~Tabas and B.~Zhang, ``Optimal l-infinity frequency control in microgrids
  considering actuator saturation,'' \emph{arXiv:1910.03720}, 2019.

\bibitem{jiang2020dynamic}
Y.~Jiang, R.~Pates, and E.~Mallada, ``Dynamic droop control in low-inertia
  power systems,'' \emph{IEEE Transactions on Automatic Control}, 2020.

\bibitem{weitenberg2018robust}
E.~Weitenberg, Y.~Jiang, C.~Zhao, E.~Mallada, C.~De~Persis, and F.~D{\"o}rfler,
  ``Robust decentralized secondary frequency control in power systems: Merits
  and tradeoffs,'' \emph{IEEE Transactions on Automatic Control}, vol.~64,
  no.~10, pp. 3967--3982, 2018.

\bibitem{weitenberg2018exponential}
E.~Weitenberg, C.~De~Persis, and N.~Monshizadeh, ``Exponential convergence
  under distributed averaging integral frequency control,'' \emph{Automatica},
  vol.~98, pp. 103--113, 2018.

\bibitem{sastry2013nonlinear}
S.~Sastry, \emph{Nonlinear systems: analysis, stability, and control}.\hskip
  1em plus 0.5em minus 0.4em\relax Springer Science \& Business Media, 2013,
  vol.~10.

\bibitem{arapostathis1982global}
A.~Arapostathis, S.~Sastry, and P.~Varaiya, ``Global analysis of swing
  dynamics,'' \emph{IEEE Transactions on Circuits and Systems}, vol.~29,
  no.~10, pp. 673--679, 1982.

\bibitem{griewank1989automatic}
A.~Griewank, ``On automatic differentiation,'' \emph{Mathematical Programming:
  recent developments and applications}, vol.~6, no.~6, pp. 83--107, 1989.

\bibitem{ortega2015generalized}
A.~Ortega and F.~Milano, ``Generalized model of vsc-based energy storage
  systems for transient stability analysis,'' \emph{IEEE transactions on Power
  Systems}, vol.~31, no.~5, pp. 3369--3380, 2015.

\bibitem{demetriou2015dynamic}
P.~Demetriou, M.~Asprou, J.~Quiros-Tortos, and E.~Kyriakides, ``Dynamic ieee
  test systems for transient analysis,'' \emph{IEEE Systems Journal}, vol.~11,
  no.~4, pp. 2108--2117, 2015.

\bibitem{nishikawa2015comparative}
T.~Nishikawa and A.~E. Motter, ``Comparative analysis of existing models for
  power-grid synchronization,'' \emph{New Journal of Physics}, vol.~17, no.~1,
  p. 015012, 2015.

\bibitem{Horn2012MA}
R.~A. Horn and C.~R. Johnson, \emph{Matrix Analysis}, 2nd~ed.\hskip 1em plus
  0.5em minus 0.4em\relax Cambridge University Press, 2012.

\end{thebibliography}
